%% file: main.tex
\documentclass[letterpaper,11pt]{article}
\usepackage[utf8]{inputenc}
\usepackage{silence}
\usepackage{amsmath, amsthm, amssymb, thm-restate}
\usepackage{algorithmicx}
\usepackage[table,xcdraw]{xcolor}
\usepackage[vlined,linesnumbered]{algorithm2e}
\usepackage{setspace}
\usepackage{mathtools}
\usepackage[numbers]{natbib}
\usepackage{comment} % enables comment environment
\usepackage[most]{tcolorbox} 
\usepackage{xfrac}
\usepackage{hyperref}
\usepackage{multirow}
\usepackage{caption}
\usepackage{bm}
\usepackage{newfloat}
\usepackage{enumitem}
\usepackage{dblfloatfix} 
\usepackage{wrapfig}

\usepackage{fancyhdr}

\SetCommentSty{mycommfont}
\let\oldnl\nl% Store \nl in \oldnl
\newcommand{\nonl}{\renewcommand{\nl}{\let\nl\oldnl}}% Remove line number for one line

\usepackage[margin=1in]{geometry}

\allowdisplaybreaks

\definecolor{mygreen}{RGB}{10,110,230}
\definecolor{myred}{RGB}{10,110,230}

\hypersetup{
     colorlinks=true,
     citecolor= mygreen,
     linkcolor= myred
}

\input{macro.tex}

\input{utils.tex}

\makeatletter
\renewcommand{\paragraph}{%
  \@startsection{paragraph}{4}%
  {\z@}{10pt}{-1em}%
  {\normalfont\normalsize\bfseries}%
}
\makeatother

\makeatletter
\patchcmd{\@algocf@start}% <cmd>
  {-1.5em}% <search>
  {0pt}% <replace>
  {}{}% <success><failure>
\makeatother

\title{Correlation Clustering Beyond the Pivot Algorithm}

\author{
Soheil Behnezhad \and 
Moses Charikar \and
Vincent Cohen-Addad \and 
Alma Ghafari \and
Weiyun Ma
}

\date{}

\begin{document}

\maketitle

\thispagestyle{empty}
\input{abstract.tex}

\input{intro.tex}

\input{sublinear.tex}

\bibliographystyle{plainnat}
\bibliography{references}
	
\end{document}

%% file: macro.tex
\renewcommand{\epsilon}{\varepsilon}

\DeclareMathOperator{\E}{\ensuremath{\normalfont \textbf{E}}}

\DeclareMathOperator{\cost}{\ensuremath{cost}}
\DeclareMathOperator{\opt}{\ensuremath{opt}}

\newcommand{\hiddencomment}[1]{}

\newcommand{\mc}[1]{\ensuremath{\mathcal{#1}}}

%% file: utils.tex
\DeclareMathOperator{\poly}{poly}

\usepackage[noabbrev,nameinlink]{cleveref}
\crefname{lemma}{Lemma}{Lemmas}
\crefname{theorem}{Theorem}{Theorems}
\crefname{property}{Property}{Properties}
\crefname{claim}{Claim}{Claims}
\crefname{definition}{Definition}{Definitions}
\crefname{observation}{Observation}{Observations}
\crefname{proposition}{Proposition}{Propositions}
\crefname{assumption}{Assumption}{Assumptions}
\crefname{line}{Line}{Lines}
\crefname{figure}{Figure}{Figures}
\creflabelformat{property}{(#1)#2#3}
\crefname{equation}{}{}
\crefname{section}{Section}{Sections}
\crefname{appendix}{Appendix}{Appendices}
\crefname{algCounter}{Algorithm}{Algorithms}
\Crefname{algCounter}{Algorithm}{Algorithms}

\newtheorem{theorem}{Theorem}
\newtheorem{lemma}{Lemma}[section]

\newtheorem{definition}[lemma]{Definition}
\newtheorem{claim}[lemma]{Claim}

\newtheorem{observation}[lemma]{Observation}
\newtheorem{remark}[lemma]{Remark}

\newtheorem{open}{Open Problem}

\definecolor{mylightgray}{RGB}{240,240,240}

\algnewcommand{\IIf}[2]{\textbf{if} #1 \textbf{then} #2}
\algnewcommand{\EndIIf}{\unskip\ \algorithmicend\ \algorithmicif}
\algnewcommand{\IElse}[1]{\textbf{else} #1}

\newenvironment{graytbox}{
\par\addvspace{0.1cm}
\begin{tcolorbox}[width=\textwidth,
                  boxsep=5pt,
                  left=1pt,
                  right=1pt,
                  top=2pt,
                  bottom=2pt,
                  boxrule=1pt,
                  arc=0pt,
                  colback=mylightgray,
                  colframe=black,
                  ]%%
}{
\end{tcolorbox}
}

\newenvironment{whitetbox}{
\par\addvspace{0.1cm}
\begin{tcolorbox}[width=\textwidth,
                  boxsep=5pt,
                  left=1pt,
                  right=1pt,
                  top=2pt,
                  bottom=2pt,
                  boxrule=1pt,
                  arc=0pt,
                  colframe=black,
                  colback=white
                  ]%%
}{
\end{tcolorbox}
}

\newcounter{algCounter}

\newenvironment{algenv}[2]{
    \begin{whitetbox}
        \refstepcounter{algCounter}
        \textbf{Algorithm~\thealgCounter:} #1
        \label{#2}
        
        \vspace{-0.2cm}
        \noindent\rule{\linewidth}{1pt}
        \begin{algorithm}[H]
}{
        \end{algorithm}
    \end{whitetbox}
}

%% file: abstract.tex
\newcommand{\apxfactor}[0]{2.997}
\newcommand{\pivot}[0]{{\normalfont \textsc{Pivot}}}
\newcommand{\modifiedpivot}[0]{{\normalfont \textsc{ModifiedPivot}}}

\begin{abstract}
{\setlength{\parskip}{0.1cm}

We study the classic correlation clustering in the dynamic setting. Given $n$ objects and a complete labeling
of the object-pairs as either “similar” or “dissimilar”, the goal is to partition the objects into
arbitrarily many clusters while minimizing disagreements with 
the labels. In the dynamic setting, an update consists of a flip
of a label of an edge.

 In a breakthrough result, [BDHSS, FOCS'19] showed how to maintain a 3-approximation with polylogarithmic update time by providing a dynamic implementation of the \pivot{} algorithm of [ACN, STOC'05]. Since then, it has been a major open problem
 to determine whether the 3-approximation barrier can be broken in the fully dynamic setting.

In this paper, we resolve this problem. Our algorithm, \modifiedpivot{}, locally improves the output of \pivot{} by moving some vertices to other existing clusters or new singleton clusters. We present an analysis showing that this modification does indeed improve the approximation to below 3. We also show that its output can be maintained in polylogarithmic time per update.

}
\end{abstract}

%% file: intro.tex
\clearpage

\section{Introduction}

Correlation clustering is a quintessential problem in data analysis, machine learning, and network science, where the task is to cluster a set of objects based on pairwise relationships. Each pair of objects is labeled as either ``similar" or ``dissimilar," and the goal is to produce clusters that best align with these labels. Formally, given $n$ vertices and their pairwise labels, the task is to partition them into arbitrarily many clusters so as to minimize the number of dissimilar labels inside clusters plus the number of similar labels that go across clusters.

 We study correlation clustering in the fully dynamic setting where each update changes the label of a pair. The goal is to maintain a good approximation of correlation clustering at all times while spending a small time per update.

 \subsection*{Background on Correlation Clustering}

The fact that correlation clustering does not require a predetermined number of clusters and that it uses both similarity and dissimilarity of the pairs make it an attractive clustering method for various tasks. Examples include image segmentation \cite{KimYNK14}, community detection \cite{ShiDELM21},  disambiguation tasks \cite{DBLP:journals/tkde/KalashnikovCMN08}, automated labeling \cite{AgrawalHKMT09,ChakrabartiKP08}, and document clustering \cite{BansalBC-FOCS02}, among others.

The correlation clustering problem was introduced by \citet*{BansalBC-FOCS02,BansalBC-ML04}, who showed that a (large) constant approximation can be achieved in polynomial time. There has been a series of polynomial-time algorithms improving the approximation ratio \cite{CharikarGW-FOCS03,AilonCN-STOC05,AilonCN-JACM08,ChawlaMSY14,Cohen-AddadLN-FOCS22,CohenAddadLLN-FOCS23}, with the current best known being the 1.437-approximation by \citet*{CCLLNV24}. It is also known that the problem is APX-hard \cite{CharikarGW-FOCS03}.

\vspace{-0.2cm}
\subsection*{Dynamic Correlation Clustering and the 3-Approximation Barrier} 
\vspace{-0.2cm}

A particularly simple and influential algorithm for correlation clustering is the \pivot{} algorithm of \citet*{AilonCN-STOC05}. The \pivot{} algorithm is remarkably simple: it picks a random vertex $v$, clusters it with vertices that are similar to $v$, then removes this cluster and recurses on the remaining vertices. 

In \cite{AilonCN-STOC05}, it was shown that \pivot{} obtains a 3-approximation for correlation clustering. Thanks to its simplicity, variants of the \pivot{} algorithm have been efficiently implemented in various models, leading to 3- or almost 3-approximations. Examples include the fully dynamic model with polylogarithmic update-time \cite{BehnezhadDHSS-FOCS19,abs-2402-15668}, constant rounds of the strictly sublinear massively parallel computations (MPC) model  \cite{Cohen-AddadLMNP21,AssadiWang-ITCS22,BehnezhadCMT-FOCS22}, a single-pass of the semi-streaming model \cite{CambusKLPU-SODA24,ChakrabartyM-NIPS23}, distributed local and congest models \cite{BehnezhadCMT-FOCS22}, and the classic RAM model where \pivot{} takes linear-time to implement.

Unfortunately, the 3-approximation analysis of the \pivot{} algorithm is tight. That is, there are various inputs on which the \pivot{} algorithm does not obtain any better than a 3-approximation.  Because of this, and the fact that all better approximations require solving large linear programs, the 3-approximation has emerged as a barrier for correlation clustering in various settings. In the case of dynamic inputs, for example, the following problem has remained open for more than 5 years since the paper of \cite{BehnezhadDHSS-FOCS19}:

\begin{open}\label{open}
    Is it possible to maintain a $3-\Omega(1)$ approximation of correlation clustering in $\poly\log n$ update-time?
\end{open}

\clearpage

We note that the problem above has been open even if one allows a much larger update-time of, say, linear in $n$.

We also note that in a very recent work \cite{combinatorialCC}, a new combinatorial algorithm was proposed for correlation clustering that obtains a much better than 3-approximation. Unfortunately, their algorithm falls short of breaking
the 3-approximation of the \pivot{} algorithm in the 
dynamic model.

\subsection*{Our Contribution} 
We show how to break the 3 bound by introducing a new
algorithm, \modifiedpivot{}, which we formalize as \cref{alg:modified-pivot}. Our algorithm modifies the output of \pivot{} by locally moving some vertices to other existing clusters or new singleton clusters. We present an analysis showing that this modification does indeed improve the approximation to below 3. Importantly, our criteria for these local moves is extremely simple. This allows the \modifiedpivot{} algorithm to be implemented as efficiently as the pivot algorithm in the dynamic setting.

\begin{graytbox}
    \begin{theorem}[\textbf{Fully Dynamic}]\label{thm:dynamic}
        There is an algorithm that maintains a $(3-\Omega(1))$-approximate correlation clustering by spending $(\poly\log n)$ time per label update against an oblivious adversary. The bounds on the update-time and the approximation hold in expectation.
    \end{theorem}
\end{graytbox}

\cref{thm:dynamic} resolves \Cref{open}.

\section{Our Techniques}

In this section, we describe the informal intuition behind our new \modifiedpivot{} algorithm. 

As standard, we model the input to correlation clustering as a graph $G=(V, E)$ with the vertex set $V$ corresponding to the objects and the edge-set $E$ representing the \underline{similar} labels. In particular, an edge $(u, v) \in E$ implies $u$ and $v$ are similar and a non-edge $(u, v) \not\in E$ implies $u$ and $v$ are dissimilar. 

It would be useful to start with the \pivot{} algorithm and discuss a few examples on which it only obtains a 3-approximation. We will then discuss how \modifiedpivot{} overcomes all of these examples and breaks the 3-approximation barrier.

With the graphic view discussed above, the \pivot{} algorithm works as follows. It iteratively picks a vertex $v$, clusters $v$ with its remaining neighbors, then removes this cluster from the graph. This continues until all vertices are removed.

\paragraph{Problem 1: \pivot{} Clusters Dissimilar Pairs.} Our first example shows a scenario where the \pivot{} algorithm, mistakenly, clusters together vertices that have very different neighborhoods. Such mistakes alone cause the \pivot{} algorithm to pay 3 times the optimum cost in these examples.

Consider a graph composed of two disjoint cliques each on $n/2$ vertices connected by one edge $(u, v)$. The optimal solution is to put the two cliques in disjoint clusters, paying only a cost of one for the edge $(u, v)$. In fact, this is exactly the clustering that \pivot{} reports so long as its first \pivot{} is not one of the endpoints of the edge $(u, v)$. However, if one of the endpoints of the edge $(u, v)$ is selected as the first pivot, then the algorithm puts $u$ and $v$ in the same cluster, paying a cost of $n-2$. The figure below illustrates this. On the left hand side, we have the optimal clustering. On the right hand side, we have the output of \pivot{} if one of the endpoints of the edge connecting the two cliques is picked as a pivot.

\begin{figure}[h]
    \centering
    \includegraphics[scale=0.75]{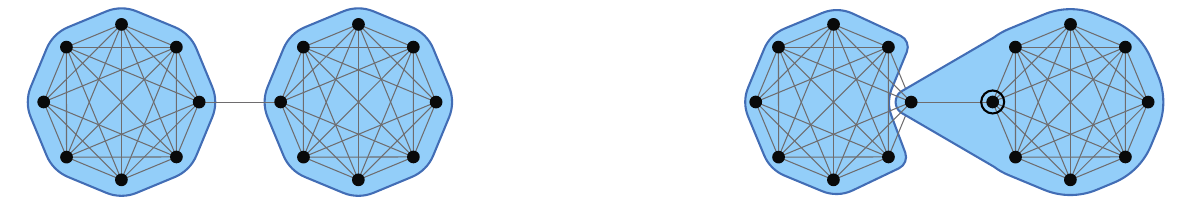}
    % \caption{Caption}
    \label{fig:enter-label}
\end{figure}

 Note that the probability that one of $u$ or $v$ is chosen as the first pivot is $2/n$, therefore, the expected cost of \pivot{} in this example is
$$
    \Pr[\text{first pivot} \not\in \{u, v\}] \cdot 1 + \Pr[\text{first pivot} \in \{u, v\}] \times (n-2) = (1-2/n) + \frac{2}{n} (n-2) \xrightarrow[n \to \infty]{} 3,
$$
which is 3 times the optimum cost.

\paragraph{Fixing Problem 1: Moving Dissimilar Neighbors to Singleton Clusters.} Our idea for fixing Problem~1 is a natural one. Whenever our \modifiedpivot{} algorithm picks a pivot $v$, we do not necessarily put all of its remaining neighbors in the cluster of $v$. Instead, if a neighbor $u$ of $v$ has a very different neighborhood than $v$, we move it to a singleton cluster. More formally, for some small constant $\delta > 0$, we first define the set $D_v$ to include neighbors $u$ of $v$ such that $|N(u) \cap N(v)| \lesssim \delta N(v)$, where $N(x)$ denotes the neighbor-set of vertex $x$ in the current graph. Note that for sufficiently small $\delta$, a vertex $u \in D_v$ has non-edges to nearly all neighbors of $v$ -- so it can only improve the cost if we move such vertices to singleton clusters.

\begin{wrapfigure}{r}{0.35\textwidth}
  \centering
  \includegraphics[scale=0.75]{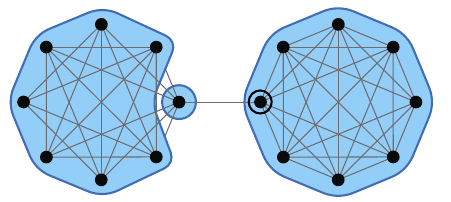}
\end{wrapfigure}
Let us now run this modified algorithm on the example of Problem~1. As before, if the first pivot is not one of the endpoints of $(u, v)$, then the algorithm returns the optimal solution with a cost of 1. But now if one of the endpoints of $(u, v)$ is picked as the first pivot, the other endpoint will move to a singleton cluster. It can be confirmed that the cost is only $n/2$ in this case. Therefore, the expected cost of the algorithm in this case will now be improved to 2 since
$$
    \Pr[\text{first pivot} \not\in \{u, v\}] \cdot 1 + \Pr[\text{first pivot} \in \{u, v\}] \times n/2 = (1-2/n) + \frac{2}{n} (n/2) \leq 2.
$$

\paragraph{Problem 2: \pivot{} Separates Similar Pairs.} It turns out that moving vertices to singleton clusters is not enough. Our next bad example for the \pivot{} algorithm shows a scenario where the \pivot{} algorithm, mistakenly, separates vertices that have to be clustered together, causing it to pay 3 times the optimum cost.

\begin{wrapfigure}{r}{0.2\textwidth}
  \centering
  \vspace{-0.5cm}
  \includegraphics[scale=0.75]{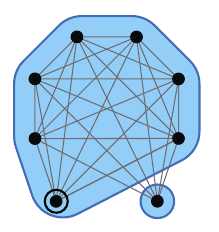}
\end{wrapfigure}
Consider a graph on $n$ vertices where all pairs are edges except one pair $(u, v)$ which is a non-edge. The optimum solution here is to put everything in the same cluster, paying only a cost of one for the non-edge. This is exactly what the \pivot{} algorithm does too, except when the first pivot chosen is one of the endpoints of the non-edge. In this case, the other endpoint of the non-edge will be put in a singleton cluster, resulting in a cost of $n-2$ as illustrated in the figure of the right hand side. 

Note that the expected cost is 3 times the optimum cost of 1 in this case too, since:
$$
    \Pr[\text{first pivot} \not\in \{u, v\}] \cdot 1 + \Pr[\text{first pivot} \in \{u, v\}] \times (n-2) = (1-2/n) + \frac{2}{n} (n-2) \xrightarrow[n \to \infty]{} 3.
$$

\paragraph{Fixing Problem 2: Moving Non-Neighbors to Pivot's Cluster.} To fix Problem~2, whenever we pick a pivot $v$, we would like to identify a set $A_v$ of non-neighbors of $v$ whose neighborhoods are similar to $N(v)$ and move them to the cluster of $v$ as well. 

The problem with doing so is that the set $A_v$ may be too large, and moving them all to the cluster of $v$ will completely change its structure. This is best described via an example. Consider a complete bipartite graph with vertex parts $V_1, V_2$ where $|V_2| \gg |V_1|$. Here the solution that puts all vertices in singleton clusters pays a cost of $|V_1| \cdot |V_2|$. Therefore, $OPT \leq |V_1| \cdot |V_2|$. But now take the first pivot $v$, which with probability $|V_2| / (|V_1| + |V_2|) = 1-o(1)$ belongs to the larger part $V_2$. Now note that all the rest of vertices in $V_2$ will have exactly the same neighborhood as $v$. Moving them all to the cluster of $v$ results in clustering all the vertices of the graph together, resulting in a cost of $\binom{|V_1|}{2} + \binom{|V_2|}{2}$ for the non-edges inside $V_1$ and $V_2$. The approximation ratio will then be at least
$$
\frac{\binom{|V_1|}{2} + \binom{|V_2|}{2}}{|V_1||V_2|} \geq \frac{\binom{|V_2|}{2}}{|V_1||V_2|} = \frac{|V_2|-1}{2|V_1|} = \omega(1).
$$
In other words, not only moving similar neighbors to the cluster of the pivot does not improve the approximation to below 3, but it worsens it to super-constant.

To fix this problem, we do not move all the vertices in $A_v$ to the cluster of $v$. Instead, we subsample some $\delta |N(v)|$ vertices in $A_v$ and only move these vertices to $v$'s cluster. It is important to note that in case $|N(v)| \ll |A_v|$, as is the case in the complete bipartite example, we only move $o(1)$ fraction of the vertices of $A_v$ to the cluster of $v$. Had this been a constant, our analysis would have been much simpler. However, we will need a much more global analysis to argue that in case $A_v$ is much larger than $N(v)$, then the output of \pivot{} is already better than 3-approximate.

\paragraph{The Final Analysis:} Up to this point, we've presented a number of instances where the approximation ratio of the \pivot{} algorithm is no better than 3. We've also explored some local improvements that would improve the approximation on these instances. What remains to show is that these local improvements do indeed beat 3-approximation on all inputs. 

Our analysis follows the standard framework of charging mistakes on {\em bad triangles}, but has an important twist. As standard, we say three vertices $\{u, v, w\}$ form a bad triangle if exactly two of the pairs $\{u, v\}$, $\{u, w\}$, $\{v, w\}$ are edges. It's important to note that regardless of how these vertices are clustered, at least one pair within a bad triangle must be incorrectly clustered. Consequently, if we can identify $\beta$ edge-disjoint bad triangles within $G$, then we can infer that the optimum cluster cost is at least $\beta$. This holds even if we identify a {\em fractional} packing of bad triangles \cite{AilonCN-STOC05}. This naturally provides a framework for analyzing the approximation ratio of correlation clustering algorithms, where the mistakes made by the algorithm are blamed on bad triangles. The crux of the analysis will then be focused on formalizing the charging scheme, i.e., which triangle to charge for each mistake and analyzing how many times each pair (edge or non-edge) is charged.

The charging scheme used for the \pivot{} algorithm by \cite{AilonCN-STOC05} is highly local, in the sense that it charges any mistake to a bad triangle involving this mistake. Our charging scheme (formalized as \cref{alg:charging}) differs from this in two crucial ways:
\begin{itemize}
    \item \textbf{Charging triangles fractionally:} Instead of charging a single bad triangle {\em integrally} for each mistake, we charge various bad triangles {\em fractionally}. In other words, there is no one-to-one mapping between our mistakes and the triangles charged. Instead, we argue that sum of the charges to the bad triangles in total is as large as the mistakes we make (\cref{cl:charge-enough}), and that sum of the charges involving each pair is not too large (\cref{clm:pair-bound}).
    \item \textbf{Charging non-local triangles:} When a pivot $v$ is picked in our \modifiedpivot{} algorithm, unlike the analysis of \cite{AilonCN-STOC05}, we do not just charge bad triangles involving the pivot. For instance, in the example of the complete bipartite graph discussed above, we charge many bad triangles that do not involve the pivot. This is the key in our analysis to show that when $A_v$ is too large compared to $C_v$, the output of \pivot{} is already good.
\end{itemize}

\section{The \modifiedpivot{} Algorithm}

Our \modifiedpivot{} algorithm is formalized below as \cref{alg:modified-pivot}. 

Let us provide some intuition about \modifiedpivot{}. Similar to \pivot{}, it iteratively picks a random pivot $v$, and based on it identifies the following sets:
\begin{itemize}
    \item $C_v$: This is the set of neighbors of $v$ still in the graph plus vertex $v$ itself. This is exactly the cluster that \pivot{} would output for $v$, but we will modify it.
    \item $D_v$: These are vertices that belong to $C_v$ but have very different neighborhood than $C_v$. Intuitively, we would like to move vertices of $D_v$ to singleton clusters instead of putting them in the cluster of $v$.
    \item $D'_v$: This is a subsample of $D_v$. Instead of moving all vertices of $D_v$ to singleton clusters, we only move vertices of $D'_v$ to singleton clusters to make sure that the cluster of $v$ does not dramatically differ from $C_v$ in size.
    \item $A_v$: These are vertices that are not adjacent to the pivot $v$, but their neighborhoods are almost the same as $C_v$. Moving each of these vertices to $C_v$ will improve our cost, provided that we do not move too many of them inside.
    \item $A'_v$: This is a subsample of $A_v$. We only move vertices of $A'_v$ to the cluster of $v$ to ensure, again, that the cluster of $v$ remains relatively close to $C_v$ in size.
    \item $A$: The set $A$ is initially empty. Whenever we pick a pivot $v$, we move all the vertices of $A_v$ to $A$. We define this set because we do not want a vertex $w$ to participate in $A_v$ and $A_u$ for two different pivots $u$ and $v$.
\end{itemize}

\begin{figure}
\begin{algenv}{The \modifiedpivot{} algorithm.}{alg:modified-pivot}
    \nonl \textbf{Parameters:} $\epsilon \in (0, \frac{1}{14}]$, $\delta \in [4\epsilon, \frac{2}{7}]$, $k \geq 1$. 

    $A \gets \emptyset$. 
    
    \While{$V \not= \emptyset$}{
        Pick a vertex $v \in V$ uniformly at random and mark it as a {\em pivot}.
        
        Let $C_v \gets \{v\} \cup N(v)$, where $N(v)$ is the set of neighbors of $v$ still in $V$.
        
        Let $D_v \gets  \{ u \mid u \in N(v) \text{ and } |N(u) \cap C_v| \leq \delta |C_v|-1\}$.

        Let $D'_v$ include $\min\{|D_v|, \lfloor \delta |C_v| \rfloor \}$ vertices of $D_v$ uniformly at random.
        
        Let $A_v := \{w \mid w \in V \setminus C_v \text{ and } w \not\in A \text{ and } |N(w) \Delta C_v| \leq \epsilon|C_v| -1 \}$.

        Let $A'_v$ include  $\min\{|A_v|, \lfloor \delta |C_v| \rfloor \}$ vertices of $A_v$ uniformly at random.

        Put each vertex of $(D'_v \setminus A) \cup (A_v \setminus A'_v)$ in a singleton cluster.

        Put all vertices of $(C_v \cup A'_v) \setminus (D'_v \cup A)$ in the same cluster.

        $A \gets A \cup A_v$. \label{line:A-update-1}
        
        Remove vertices of $C_v$ from $V$. 
        
        \nonl (We emphasize that even though vertices in $A_v$ get clustered here, they are not removed from $V$ in this step and so can be picked as pivots later on.)
    }
\end{algenv}
\end{figure}

The following observation shows that the output of \modifiedpivot{} is a valid clustering. What remains is to analyze its approximation ratio, which we do in \cref{sec:analysis}.

\begin{observation}
    The output of \cref{alg:modified-pivot} is always a valid clustering. That is, each vertex belongs to exactly one cluster of the output with probability 1.
\end{observation}

\begin{proof}
    First, observe that for every $i$, the set of vertices removed from $V$ in the first $i$ iterations of \Cref{alg:modified-pivot} is identical to the set of vertices clustered in the first $i$ iterations of \pivot{} under the same random coin tosses. Since \Cref{alg:modified-pivot} only removes a vertex from $V$ if it has been clustered (either in the same iteration or an earlier iteration), this means that every vertex gets clustered at some point in \Cref{alg:modified-pivot}. Moreover, if a vertex is clustered in some iteration of \Cref{alg:modified-pivot}, then it is either removed from $V$ or added to the set $A$ at the end of that iteration. Since \Cref{alg:modified-pivot} never clusters a vertex that has been removed from $V$ or is already in $A$, this means that a vertex cannot be clustered more than once. Thus \cref{alg:modified-pivot} always outputs a valid clustering.
\end{proof}

\section{Analysis of {\normalfont \modifiedpivot{}}}\label{sec:analysis}

In this section, we analyze the approximation ratio of the \modifiedpivot{} algorithm, proving the following theorem:

\begin{theorem}\label{thm:approx}
    The clustering output by the \modifiedpivot{} algorithm has cost at most $\apxfactor$ times the optimal cost in expectation.
\end{theorem}

\begin{remark}
    We note that we have not tried to optimize the approximation ratio in \cref{thm:approx} as our main contribution is the qualitative result that the $3$-approximation is not the ``right'' bound for correlation clustering across various settings.
\end{remark}

The analysis still fits into the framework of charging {\em bad triangles} as in the original 3-approximation analysis of the \pivot{} algorithm \cite{AilonCN-JACM08}. However, the triangles charged in our analysis are very different from \cite{AilonCN-JACM08}. We first provide the needed background on charging bad triangles in \cref{sec:badtriangles}, then formalize our analysis using this framework in \cref{sec:chargingscheme}.

\subsection{Background on Charging Bad Triangles}\label{sec:badtriangles}

Let us first overview the framework of {\em charging bad triangles} \cite{AilonCN-JACM08}. We say three distinct vertices $\{a, b, c\}$ in $V$ form a {\em bad triangle} if exactly two of the pairs $\{a, b\}, \{a, c\}, \{b, c\}$ belongs to $E$. Let $BT$ be the set of all bad triangles in the graph.

\begin{definition}
Let $\mc{A}$ be a (possibly randomized) algorithm for correlation clustering. We say an algorithm $\mc{S}$ is a {\em charging scheme of width $w$ for $\mc{A}$} if for every given output clustering $\mc{C}$ of $\mc{A}$ and every bad triangle $t \in BT$, algorithm $\mc{S}$ specifies a real $y_t \geq 0$ such that:
\begin{enumerate}
    \item $\sum_{t} y_t \geq \cost(\mc{C})$.
    \item For every distinct $u, v \in V$ (which may or may not belong to $E$), it holds that
    $$
    \E_{\mc{A}}\left[\sum_{t \in BT: u, v \in t}  y_t\right] \leq w.
    $$
\end{enumerate}
\end{definition}

The following lemma shows why charging schemes are useful.

\begin{lemma}\label{lem:width-gives-apx}
Let $\mc{A}$ be any (possibly randomized) correlation clustering algorithm. If there exists a charging scheme of width $w$ for $\mc{A}$, then for the clustering $\mc{C}$ produced by $\mc{A}$,
$$
    \E_{\mc{A}}[\cost(\mc{C})] \leq w \cdot \opt(G).
$$
\end{lemma}

\Cref{lem:width-gives-apx} is a standard result in the literature and follows from a simple primal dual argument. See for example \cite{AilonCN-STOC05} or \cite[Appendix~C]{BehnezhadCMT-FOCS22} for its proof.

\subsection[Our Charging Scheme for Modified Pivot Algorithm]{Our Charging Scheme for\cref{alg:modified-pivot}}\label{sec:chargingscheme}

The following \cref{alg:charging} formalizes our charging scheme for \modifiedpivot{}. \cref{alg:charging} proceeds exactly like \modifiedpivot{} and defines all the sets used by \modifiedpivot{} in forming its clusters. However, instead of returning a clustering, \cref{alg:charging} returns a charge $y_{t} \geq 0$ for each bad triangle $t \in BT$.

In \cref{sec:charge-enough} we show that \cref{alg:charging} charges as many bad triangles as the cost paid by \modifiedpivot{}. We then prove in \cref{sec:width} that \cref{alg:charging} has width at most $\apxfactor$. Combining these lemmas and plugging them into \cref{lem:width-gives-apx} proves \cref{thm:approx} that \modifiedpivot{} obtains a \apxfactor{}-approximation.

\begin{figure}
\begin{algenv}{The charging scheme for analyzing \modifiedpivot{}.}{alg:charging}
    \nonl \textbf{Parameters:} $\epsilon \in (0, \frac{1}{14}]$, $\delta \in [4\epsilon, \frac{2}{7}]$, $k \geq 1$.

    $A \gets \emptyset$. 
    
    \While{$V \not= \emptyset$}{
        Pick a vertex $v \in V$ uniformly at random and mark it as a {\em pivot}.
        
        Let $C_v \gets \{v\} \cup N(v)$, where $N(v)$ is the set of neighbors of $v$ still in $V$.

        Let $D_v \gets  \{ u \mid u \in N(v) \text{ and } |N(u) \cap C_v| \leq \delta |C_v|-1\}$. 
        % \sbcomment{Changed from $ |N(u) \cap N(v)| \leq \delta |N(v)|$.}

        Let $D'_v$ include $\min\{|D_v|, \lfloor \delta |C_v| \rfloor \}$ vertices of $D_v$ uniformly at random.
        % \sbcomment{Changed ceilings to floors.}
        
        Let $A_v := \{w \mid w \in V \setminus C_v \text{ and } w \not\in A \text{ and } |N(w) \Delta C_v| \leq \epsilon|C_v| -1 \}$.

        Let $A'_v$ include  $\min\{|A_v|, \lfloor \delta |C_v| \rfloor \}$ vertices of $A_v$ uniformly at random.

        \For{every $(u, w) \not\in E$ such that $u, w \in C_v$}{
            \If{$u \not\in D'_v$ and $w \not\in D'_v$}{
                $y_{(v, u, w)} \gets 1$. \label{line:C1}
            }\Else{
                $y_{(v, u, w)} \gets 2 \delta/(1-\frac{3}{2}\delta)$.
                
                 \label{line:C2}
                
            }
        }

        \If{$|A_v| \leq k|C_v|$}{
            \For{every $(u, w) \in E$ where $u \in C_v$, $w \in V \setminus C_v$}{
                \If{$w \in A$}{
                    Do not charge a new triangle for $(u, w)$.
                   
                }\Else{
                    \IIf{$w \not\in A_v$}{$y_{(v, u, w)} \gets 1$.\label{line:C8}} 
                    
                    \If{$w \in A_v$}{
                        \If{$w \in A'_v$}{
                        $y_{(v, u, w)} \gets \delta.$ \label{line:C3}
                        }
                        \Else{$y_{(v, u, w)} \gets 1+\frac{\epsilon}{1 - \epsilon}$.  \label{line:C4} }
                    }
                }
            }
        }
        
        \If{$|A_v| > k|C_v|$}{
            \For{every mistake $(u, w) \in E$ where $u \in C_v$,  $w \in V \setminus C_v$}{
                \If{$w \in A$}{
                    Do not charge a new triangle for $(u, w)$.
                }\Else{
                    \IIf{$w \not\in A_v$}{$y_{(v, u, w)} \gets 1$. \label{line:C5}}
                
                    \IIf{$w \in A_v$}{$y_{(v, u, w)} \gets 1- \frac{\epsilon}{1 - \epsilon}$. \label{line:C6}}
                }
                
            }

            \For{every bad triangle $(u, w, x)$ such that $u \in N(v)$, $w \in A_v$, $x \in A_v$, $(w, x) \not\in E$, $(u, w) \in E$, and $(u, x) \in E$}{
                $y_{(u, w, x)} \gets \frac{5\epsilon / (1-\epsilon)}{|A_v| - 1}$. \label{line:C7} 
            }
            
        }

        $A \gets A \cup A_v$. \label{line:A-update-2}
        
        Remove vertices of $C_v$ from $V$.
    }

    Return $y$.
\end{algenv}
\end{figure}

\subsection[Our Charging Scheme Charges Enough Bad Triangles]{\cref{alg:charging} Charges Enough Bad Triangles}\label{sec:charge-enough}

In this section, we show that \cref{alg:charging} charges enough bad triangles.

\begin{lemma}\label{cl:charge-enough}
    Let $y$ be the vector of charges returned by \cref{alg:charging} and let $\mc{C}$ be the corresponding clustering returned by \modifiedpivot{} (\cref{alg:modified-pivot}). Then it holds that $$\sum_{t \in BT} y_t \geq \cost(\mc{C}).$$
\end{lemma}
\begin{proof}
    We prove by induction that at the end of every iteration $i$ of the while loop, $\sum_{t \in BT} y_t$ upper bounds the number of mistakes made by \cref{alg:modified-pivot} so far. Clearly this holds for the base case $i = 0$. 
    
    Now consider iteration $i \geq 1$. The set of vertices newly clustered in this iteration is $C_v \cup A_v \setminus A$. (To avoid ambiguity, any mention of the set $A$ during iteration $i$ in this proof specifically refers to its state before it is updated by $A_v$ in \Cref{line:A-update-1} of \cref{alg:modified-pivot} or \Cref{line:A-update-2} of \cref{alg:charging}.) To prove the inductive step, it suffices to show that the number of mistakes newly made by \cref{alg:modified-pivot} in iteration $i$, which are precisely the mistakes that have at least one endpoint in $C_v \cup A_v$ and no endpoint in $A$, are upper bounded by the total amount of charge to bad triangles in \cref{line:C1,line:C2,line:C8,line:C3,line:C4,line:C5,line:C6,line:C7} in this iteration. Note that each of these mistakes $(x, z)$ satisfies \emph{exactly} one of the following conditions:

    \begin{enumerate}[label=(\arabic*), ref=\arabic*, leftmargin=*, align=left]
        \item $(x, z) \notin E$ and $x, z \in C_v \setminus D'_v$.
        \item $(x, z) \in E$, $x \in D'_v$ and $z \in C_v \cup A'_v$.
        \item $(x, z) \in E$, $x \in C_v$ and $z \in V \setminus (C_v \cup A_v \cup A)$.
        \item $(x, z) \notin E$ and $x, z \in A'_v$.
        \item Either $(x, z) \in E$, $x \in A'_v$ and $z \in V \setminus (C_v \cup A'_v)$, or $(x, z) \notin E$, $x \in A'_v$ and $z \in C_v \setminus D'_v$.
        \item $(x, z) \in E$, $x \in A_v \setminus A'_v$ and $z \in C_v$.
        \item $(x, z) \in E$, $x \in A_v \setminus A'_v$ and $z \in V \setminus 
        (C_v \cup A'_v)$.
    \end{enumerate}

    We refer to the mistakes that satisfy condition ($j$) as Type ($j$) mistakes. Let $c_j$ denote the number of mistakes of Type ($j$) and let $y_l$ denote the total amount of charge to bad triangles in Line $l$ of \cref{alg:charging} in iteration $i$. We now prove the following statements (a)-(d) one by one, which collectively imply the inductive step:
    \begin{enumerate}[label= (\alph*), ref=\alph*, leftmargin=*, align=left]
    \item $c_1 \leq y_{11}$.

    To see this holds, we observe that each Type (1) mistake $(x, z)$ where $(x, z) \notin E$ and $x, z \in C_v \setminus D'_v$ corresponds to a bad triangle $(v, x, z)$ that is charged by 1 in \cref{line:C1}.
    
    \item $c_2 \leq y_{13}$.
        
    The total number of Type (2) mistakes $(x, z)$ where $(x, z) \in E$, $x \in D'_v$ and $z \in C_v \cup A'_v$ is at most 
     $$\sum_{x\in D'_v}{(|N(x) \cap C_v| + |A'_v|)} \leq |D'_v|\left(\delta |C_v|-1 + \lfloor \delta |C_v| \rfloor\right) \leq 2\delta |D'_v| |C_v|.$$ 
    On the other hand, the number of pairs $(u, w)\notin E$ such that $u, w \in C_v$ and at least one of $u$ or $w$ is in $D'_v$, or equivalently, the number of bad triangles $(v,u,w)$ that are charged in \cref{line:C2}, is equal to 
    \begin{align*}
     & \sum_{u\in D'_v}{\left( \left |(C_v \setminus D'_v) \setminus N(u)\right | + \frac{1}{2}\left |D'_v \setminus (N(u) \cup \{u\})\right |\right)} \\
     = & \sum_{u\in D'_v}{\left( \left| C_v \setminus (N(u) \cup \{u\})\right| - \frac{1}{2}\left|D'_v \setminus (N(u) \cup \{u\})\right| \right)} \\
     \geq & \left(\sum_{u\in D'_v}{(|C_v| - |C_v \cap N(u)| - 1)}\right) - {|D'_v| \choose 2} \\
     \geq & |D'_v|\left(|C_v| - (\delta|C_v| - 1) - 1 - \genfrac{}{}{}{}{1}{2}\left( \lfloor\delta|C_v|\rfloor - 1\right) \right) \\
     \geq & \left(1 - \genfrac{}{}{}{}{3}{2}\delta\right)|D'_v||C_v|.
     \end{align*}
    Thus the total amount of charge in \cref{line:C2} is at least $$\frac{2\delta}{1-\frac{3}{2}\delta} \left(1 - \frac{3}{2}\delta\right)|D'_v||C_v| = 2\delta |D'_v||C_v|,$$
    which upper bounds the total number of Type (2) mistakes.
    
    \item \textit{If $|A_v| \leq k|C_v|$, then $c_3 \leq y_{19}$, $c_4 + c_5 \leq y_{22}$, and $c_6 + c_7 \leq y_{24}$.}

    In the case of $|A_v| \leq k|C_v|$, \Cref{alg:charging} charges in \cref{line:C8,line:C3,line:C4}. We show the three inequalities separately. 
    
    To see that $c_3 \leq y_{19}$, we observe that each Type (3) mistake $(x, z)$ where $(x, z) \in E$, $x \in C_v$ and $z \in V \setminus C_v \setminus A_v \setminus A$ corresponds to a bad triangle $(v, x, z)$ that is charged by 1 in \cref{line:C8}.

    Next, we show $c_4 + c_5 \leq y_{22}$. The total number of Type (4) mistakes $(x, z)$ where $(x, z) \notin E$ and $x, z \in A'_v$ is at most $${|A'_v|\choose 2} = \frac{1}{2}|A'_v|(|A'_v| - 1) \leq \frac{1}{2}|A'_v|(\lfloor \delta |C_v| \rfloor - 1) \leq \frac{\delta}{2}|A'_v||C_v|.$$ For type (5) mistakes $(x, z)$ where either $(x, z) \in E$, $x \in A'_v$ and $z \in V \setminus C_v \setminus A'_v$, or $(x, z) \notin E$, $x \in A'_v$ and $z \in C_v \setminus D'_v$, note that in both cases we have $z \in N(x) \Delta C_v$. Thus the total number of Type (5) mistakes is at most $$\sum_{x\in A'_v}|N(x) \Delta C_v| \leq |A'_v| (\epsilon |C_v| - 1) \leq \epsilon |A'_v||C_v|.$$
    On the other hand, the number of pairs $(u, w)\in E$ such that $u \in C_v$ and $w\in A'_v$, or equivalently, the number of bad triangles $(v,u,w)$ that are charged in \cref{line:C3}, is equal to $$\sum_{w\in A'_v}{\left | N(w) \cap C_v \right |} = \sum_{w\in A'_v}{\left | C_v \setminus (N(w) \Delta C_v) \right |} \geq |A'_v|(|C_v| - (\epsilon|C_v| - 1)) \geq (1-\epsilon)|A'_v||C_v|.$$
    Thus the total amount of charge in \cref{line:C3} is at least $$\delta (1-\epsilon)|A'_v||C_v| \geq (\delta - \epsilon)|A'_v||C_v| \geq \left(\frac{\delta}{2} + \epsilon\right)|A'_v||C_v|,$$
    where the last two inequalities follows from $\epsilon \in (0, \frac{1}{14}]$ and $\delta \in [4\epsilon, \frac{2}{7}]$. This upper bounds the total number of Type (4) and (5) mistakes.

    Last, we show $c_6 + c_7 \leq y_{24}$. Note that each Type (6) mistake $(x, z)$ where $(x, z) \in E$, $x \in A_v \setminus A'_v$ and $z \in C_v$ corresponds to a bad triangle $(v, z, x)$ that is charged by $1 + \frac{\epsilon}{1 - \epsilon}$ in \cref{line:C4}. For each such $(v, z, x)$, we allocate a charge of 1 to cover Type (6) mistakes. It remains to show that the sum of remaining charge of $\frac{\epsilon}{1 - \epsilon}$ to each of these triangles in \cref{line:C4} is sufficient to cover Type (7) mistakes as well. To that end, let us count the number of bad triangles charged in \cref{line:C4}, which is 
    \begin{flalign*}
        \sum_{w\in A_v \setminus A'_v}{|N(w) \cap C_v|} &= \sum_{w\in A_v \setminus A'_v}{|C_v \setminus (N(w) \Delta C_v)|}\\
        &\geq |A_v \setminus A'_v|(|C_v| - (\epsilon |C_v| - 1))\\
        &\geq (1 - \epsilon)|A_v \setminus A'_v||C_v|.
    \end{flalign*}
    
    Thus the total amount of remaining charge we can allocate for Type (7) mistakes is at least $$\frac{\epsilon}{1 - \epsilon}(1 - \epsilon)|A_v \setminus A'_v||C_v| = \epsilon |A_v \setminus A'_v||C_v|.$$
    We now show that the total number of Type (7) mistakes does not exceed this amount. Indeed, the total number of Type (7) mistake $(x, z)$ where $(x, z) \in E$, $x \in A_v \setminus A'_v$ and $z \in V \setminus C_v \setminus A'_v$ is at most $$\sum_{x\in A_v \setminus A'_v}{|N(x) \Delta C_v|} \leq |A_v \setminus A'_v|(\epsilon |C_v| - 1) \leq \epsilon|A_v \setminus A'_v||C_v|.$$
        
    \item \textit{If $|A_v| > k|C_v|$, then
    $c_3 \leq y_{30}$ and $c_4 + c_5 + c_6 + c_7 \leq y_{31} + y_{33}$.}

    In the case of $|A_v| > k|C_v|$, \Cref{alg:charging} charges in \cref{line:C5,line:C6,line:C7}.
   
   We first show $c_3 \leq y_{30}$. To see this holds, we observe that each Type (3) mistake $(x, z)$ where $(x, z) \in E$, $x \in C_v$ and $z \in V \setminus C_v \setminus A_v \setminus A$ corresponds to a bad triangle $(v, x, z)$ that is charged by 1 in \cref{line:C5}.

    We then show $c_4 + c_5 + c_6 + c_7 \leq y_{31} + y_{33}$. Recall that in the case of $|A_v| \leq k|C_v|$, we showed $c_4 + c_5 + c_6 + c_7 \leq y_{22} + y_{24}$. Suppose for a moment that \Cref{alg:charging} had charged each bad triangle $(v, u, w)$ in \cref{line:C6} by $\max{(\delta, 1 + \frac{\epsilon}{1 - \epsilon}}) = 1 + \frac{\epsilon}{1 - \epsilon}$. Then by the exactly same argument as we had for the case of $|A_v| \leq k|C_v|$, we could show that $c_4 + c_5 + c_6 + c_7 \leq y_{31}$ holds as well. However, in reality, \Cref{alg:charging} only charges an amount of $(1 - \frac{\epsilon}{1 - \epsilon})$ to each bad triangle $(v, u, w)$ in \cref{line:C6}. Since there are at most $|A_v|$ choices for $w \in A_v$ and at most $(|C_v| - 1)$ choices for $u \in C_v \setminus \{v\}$, this results in a total charge deficit of at most $\frac{2\epsilon}{1 - \epsilon}|A_v|(|C_v| - 1)$. 

    To cover this deficit, we show that $y_{33} \geq \frac{2\epsilon}{1 - \epsilon}|A_v|(|C_v|-1)$. To that end, we need to show that \Cref{alg:charging} charges enough bad triangles in \cref{line:C7}. The total number of triplets $(u, w, x)$ such that $u \in N(v)$ and $w,x \in A_v$ is equal to $${|A_v| \choose 2}(|C_v| - 1).$$ Note that each pair $(u, w)$ where $u \in N(v)$ and $w\in A_v$ can appear in at most $|A_v| - 1$ such triplets, and each pair $(w, x)$ where $w, x \in A_v$ can appear in at most $|C_v| - 1$ such triplets. Thus the total number of such triplets $(u, w, x)$ that do not satisfy the condition in Line 32 and are not charged in \cref{line:C7} is at most 
    \begin{align*}
    & \sum_{\substack{(u,w):(u,w)\notin E, \\ u \in N(v), \\ w\in A_v}}{(|A_v| - 1)} + \sum_{\substack{(w,x): (w,x)\in E, \\ w, x\in A_v}}{(|C_v| - 1)}\\
    = &\sum_{w\in A_v} \left( \sum_{u\in C_v \setminus N(w)}(|A_v| - 1) + \frac{1}{2}\sum_{x\in N(w) \cap A_v}(|C_v| - 1) \right) \\
    \leq & \sum_{w\in A_v}|N(w) \Delta C_v|\max\left(|A_v|-1, \frac{1}{2}(|C_v| - 1)\right) \\
    \leq & |A_v|(\epsilon|C_v| - 1)(|A_v| - 1),
    \end{align*}
    where the last inequality follows from $|A_v| > k |C_v|$ and $k \geq 1$.
    Thus the number of bad triangles charged in \cref{line:C7} is at least $${|A_v| \choose 2}(|C_v| - 1) - |A_v|(\epsilon|C_v| - 1)(|A_v| - 1) \geq (\frac{1}{2} - \epsilon)|C_v||A_v|(|A_v| - 1).$$
    Thus the total amount of charge in \cref{line:C7} is at least 
    \begin{align*}
    & \frac{5\epsilon/(1-\epsilon)}{|A_v|-1}(\frac{1}{2} - \epsilon)|C_v||A_v|(|A_v| - 1) 
    \geq \frac{5\epsilon(1/2-\epsilon)}{1-\epsilon}|A_v||C_v| 
    \geq \frac{2\epsilon}{1-\epsilon}|A_v||C_v|,
    \end{align*}
    where the last inequality follows from $\epsilon \leq \frac{1}{14}$. This is sufficient to cover the total deficit of at most $\frac{2\epsilon}{1 - \epsilon}|A_v|(|C_v| - 1)$ from \cref{line:C6}.
    \end{enumerate}
    We have proved statements (a)-(d) for iteration $i$. By induction, the proof is complete.
\end{proof}

\subsection[Our Charging Scheme Has Width Smaller than 3]{\cref{alg:charging} Has Width Smaller than 3}\label{sec:width}

In this section, we prove that \cref{alg:charging}, for any fixed pair of vertices, charges at most $\apxfactor$ bad triangles involving them in expectation. This upper bounds the width of \cref{alg:charging} by $\apxfactor$, and thus combined with \cref{cl:charge-enough} and \cref{lem:width-gives-apx} proves that \cref{alg:modified-pivot} obtains a $\apxfactor$-approximation.

Let us for every pair $(a, b)$ of the vertices use $y_{(a, b)} := \sum_{t \in BT: a, b \in t}  y_t$ to denote the total charges to the bad triangles involving both $a$ and $b$. Our main result of this section is the following lemma.

\begin{lemma}\label{clm:pair-bound}
    Let $y$ be the charges returned by \cref{alg:charging}. For every pair $(a, b)$ of vertices,
    $$
    \E_{\mc{A}}\left[y_{(a, b)}\right] \leq \apxfactor.$$
\end{lemma}

In order to prove \cref{clm:pair-bound}, we start with a number of useful observations. When we say a pair $(a, b)$ of vertices is charged in \cref{alg:charging}, we mean that \cref{alg:charging} charges some bad triangle involving $(a, b)$. 

\begin{observation}\label{obs:pivot-in-triangles-except-C7}
    Except for the bad triangles charged in \cref{line:C7} of \cref{alg:charging}, whenever a bad triangle $t$ is charged in \cref{alg:charging}, the pivot $v$ chosen in that iteration must be part of $t$.
\end{observation}
\begin{proof}
    Follows directly from the description of \cref{alg:charging}.
\end{proof}

\begin{observation}\label{obs:edge-charged-when-deleted}
    Any edge $(a, b) \in E$ is charged in at most one iteration of \cref{alg:charging}. Any non-edge $(a, b) \notin E$ is charged in at most two iterations of \cref{alg:charging}, and in particular, is charged in at most one iteration if none of the charges involving it take place in \cref{line:C7}.
\end{observation}
\begin{proof}
    First, as shown in \cref{obs:pivot-in-triangles-except-C7}, except for when a triangle is charged in \cref{line:C7} of \cref{alg:charging}, the pivot $v$ must be part of the bad triangle. This means that either $a$ or $b$ should be chosen as the pivot $v$ or at least one of them must be adjacent to $v$. In either case, at least one of $u$ or $v$ gets removed from $V$ in iteration $i$. Note that, at least one of $(a,b)$ is corresponded to either $u$ or $v$, as a result of this at least one of the endpoints of $(a,b)$ is removed from $V$, and therefore, $(a,b)$ won't be charged again.
    
    Now, if $(a,b) \in E$, consider the case where a bad triangle $(u, w, x)$ is charged in \cref{line:C7}. In this case, $u \in C_v$ gets removed from $V$ in this iteration but $w$ and $x$ remain in $V$. Crucially, observe that the two edges of this bad triangle, which are $(u, v)$ and $(u, w)$, are both adjacent to $u$. Therefore, in this case too, any edge that is part of a charged bad triangle has at least one endpoint removed.  Note that, $(a,b)$ is corresponded to either $(w,u)$ or $(x,u)$. This means after charging $(a,b)$ in \cref{line:C7} of \cref{alg:charging}, we remove at least one of $(a,b)$ from $V$, and consequently, we will not charge $(a,b)$ in any future iterations.
    
    If $(a, b)\notin E$, then it can be involved in multiple bad triangles $(u, w, x)$ charged in \cref{line:C7} of \cref{alg:charging} in one iteration. However, we will not be charging this non-edge in \cref{line:C7} again in any future iteration of \cref{alg:charging}. This is because we will be appending $w$ and $x$ to the set $A$, which means that we will not be charging this pair as a member of $A_{v'}$ for a pivot $v'$ in a future iteration. However, we might still charge this non-edge $(a, b)$ in one more future iteration in a single line other than \cref{line:C7}.
\end{proof}

Let us group the bad triangles charged in \cref{alg:charging} in iteration $i$ based on the position of the pivot. Note that each charging line in the algorithm processes a particular kind of bad triangle. We define these sets based on whether a bad triangle includes a pivot $v$ or not, and if yes what the adjacency state of $v$ is.

\begin{definition}\label{def:triangle-types}
    Let $v$ be the pivot chosen in some iteration $i$ of \cref{alg:charging}. 
     Let $X_v$ be the set of bad triangles $t$ in the graph of iteration $i$ which involve the pivot $v$ and $v$ is adjacent to the other two vertices in $t$. Let $Y_v$ be the set of bad triangles $t$ in the graph of iteration $i$ which involve the pivot $v$ and $v$ is adjacent to exactly one other vertex of $t$. Finally, let $Z_v$ be the set of all bad triangles in the graph of iteration $i$ that are charged in this iteration but do not include the pivot $v$.
\end{definition}

Now, we investigate the charges for each type of bad triangles.
\begin{observation}\label{obs:triangle-types}
    By the assumption that pivot $v$ was picked in iteration $i$ of \cref{alg:charging} it holds that:
    \begin{enumerate}
        \item Any $t \in X_v$ is charged by either one of the \cref{line:C1,line:C2} and therefore is charged at most by 1.
        \item  Any $t \in Y_v$ is charged by either one of the \cref{line:C8,line:C3,line:C4,line:C5,line:C6} and therefore is charged at most by $1+\frac{\epsilon}{1 - \epsilon}$.
        \item Any $t \in Z_v$ is charged $\frac{5\epsilon / (1-\epsilon)}{|A_v| - 1}$  by only \cref{line:C7}.
    \end{enumerate}
\end{observation}

\begin{proof}
    We prove the three cases one by one below.
    \begin{enumerate}
         \item Note that followed by the charging scheme in \cref{line:C1,line:C2} of \cref{alg:charging} we charge bad triangles including a pivot $v$ and its neighbors $u$ and $w$ in iteration $i$ of the algorithm. That is by description, all the bad triangles in set $X_v$. Note that the charge of $t$ is bounded by maximum charge of \cref{line:C1,line:C2} that is equal to $\max ( 1, \frac{2\delta}{ 1- \frac{3}{2}\delta}) $. Note that by the choice of parameter $\delta \leq \frac{2}{7}$ in \cref{alg:modified-pivot}, we have $\frac{2\delta}{ 1- \frac{3}{2}\delta} \leq 1$, and therefore, $\max ( 1, \frac{2\delta}{ 1- \frac{3}{2}\delta}) = 1$.
        
        \item The structure of triangles in $Y_v$, is also the same as our charging cases in \cref{line:C8,line:C3,line:C4,line:C5,line:C6}. Note that we charge bad triangles in iteration $i$ including the pivot $v$, vertex $u\in C_v$ and, $w \in V\setminus C_v$. In this case, each triangle is charged at most by $\max (\delta, 1, 1-\frac{\epsilon}{1 - \epsilon} , 1+\frac{\epsilon}{1 - \epsilon}) = 1 + \frac{\epsilon}{1 - \epsilon}$.
        \item Finally, by description any bad triangle in set $Z_v$ is charged by  \cref{line:C7}, we charge each triangle in this set by $\frac{5\epsilon / (1-\epsilon)}{|A_v| - 1}$. 
        \end{enumerate}
    This completes the proof.
\end{proof}

\begin{definition}
    We define $N(a)$ in iteration $i$ of  \cref{alg:modified-pivot} as the set of the remaining neighbors of $a$ in $V$.
\end{definition}
\begin{definition}
Note that, for analyzing different bad triangles containing vertices $a$ and $b$ we need to define the sets where the third vertex $c$ is chosen from. Confirm that vertex $c$ should be in a neighborhood of $a$ or $b$. We define the following sets based on adjacency of vertex $c$ to $a$, $b$, or, both:
\begin{flalign*}
    N_a &:= N(a) \setminus (N(b) \cup b),\\
    N_b &:= N(b) \setminus (N(a)\cup a),\\
    N_{a,b} &:= (N(a) \cap N(b)) \setminus \{a,b\}.
\end{flalign*}
Note that these sets are defined based on the vertices remaining in the graph in iteration $i$ of \cref{alg:modified-pivot}.
\end{definition}
\begin{definition}
                Let us define $y_{(a, b), S}$ as the sum of the charges returned from \cref{alg:charging} for any bad triangle $t$ containing vertices $(a,b,c)$ such that $c \in S$. That is, we define
                $$y_{(a, b), S} := \sum_{t \in BT: a,b,c \in t, c \in S}y_t.$$
\end{definition}

To prove \cref{clm:pair-bound}, we need to separate the analysis into two parts. Particularly, the analysis of the edges and non-edges is different, this is because the charging scheme is not symmetric with respect to the adjacency of two vertices. 

\subsubsection{Width Analysis for Edges}\label{sec:width-e}

\begin{claim}\label{lem:edge-bound}
    For any $(a,b) \in E$ we have:
    \begin{enumerate}
        \item $\E[y_{(a, b), N_{a,b} }] =  0$,
        \item $\E[y_{(a, b)} \mid v \in (N(a) \Delta N(b) ) \setminus \{a,b\}] \leq 1+\frac{4\epsilon}{1 - \epsilon}$,
        \item $\E[y_{(a, b), N_a} \mid v =a ] \leq |N_a|$,
        \item $\E[y_{(a, b), N_b} \mid v =a ] \leq (1+\frac{\epsilon}{1-\epsilon})|N_b|$.
    \end{enumerate}
    
\end{claim}

\begin{proof}
    Here we prove each statement separately.

    \begin{enumerate}
        \item We do not charge $t$ in \cref{alg:charging} if $c \in N(a) \cap N(b)$, as $t$ will not form a bad triangle.

        \item 
        In this case, $v$ is adjacent to exactly one of $a$ or $b$ due to the conditional event $v \in (N(a) \Delta N(b) ) \setminus \{a,b\}$. Let us assume without loss of generality that $v$ is adjacent to $a$. We consider the following three cases which cover all possibilities:
        \begin{itemize}
                \item $|A_v| \leq k|C_v|$:
                
                Confirm that, \cref{alg:charging} implies that in this setting we will only charge bad triangle $t = (a,b,v) \in Y_v$ .  The charges include \cref{line:C8,line:C3,line:C4}. The maximum charge for $t$ is $1+ \frac{\epsilon}{1 - \epsilon}$.

                \item $|A_v| > k|C_v|$ and $b \not\in A_v$: In this case if $b \notin A_v$ the only charge that applies to bad triangles $t$ including $(a,b)$ is the charge in \cref{line:C5}, this bounds the charge of $(a,b)$ by 1 for each choice of the pivot.

                \item $|A_v| > k|C_v|$ and $b \in A_v$:

                In this case, there are two types of bad triangles that involve $(a, b)$: bad triangles of type $(a,b,c) \in  Z_v$ charged in \cref{line:C7} and those of type $(a,b,v) \in Y_v$ charged in \cref{line:C6}.  Note that we charge $(a,b,v)$ in \cref{line:C6} by $1-\frac{\epsilon}{1 - \epsilon}$. In \cref{line:C7}, for any vertex $x$ such that $x \in N_a \cap A_v$ we charge $(a,b,x)$ by $\frac{5\epsilon / (1-\epsilon)}{|A_v| - 1}$. Since $x \in A_v$, there are at most $|A_v|-1$ choices of $x$ and so the total charge from such triangles involving $(a, b)$ is at most $\frac{5\epsilon / (1-\epsilon)}{|A_v| - 1} \cdot (|A_v|-1) = \frac{5\epsilon}{1 - \epsilon}$. Combined with the charge of $1-\frac{\epsilon}{1 - \epsilon}$ incurred in \cref{line:C6}, this sums up to at most a charge of $1+\frac{4\epsilon}{1 - \epsilon}$.
                \end{itemize}
            \item In this case, since $v = a$ and $a$ is adjacent to both endpoints of any bad triangle counted in $y_{(a, b), N_a}$, all such bad triangles belong to $X_v$ by \cref{def:triangle-types}. By \cref{obs:triangle-types}, any bad triangle in $X_v$ is charged at most by $1$. Since there are at most $|N_a|$ choices of the third vertex in bad triangles counted in $y_{(a, b), N_a}$ and each is charged by at most $1$ as discussed earlier, the total charges sum up to at most $|N_a|$.
            \item In this case, since $v =a $  the pivot is adjacent to $b$ and is not adjacent to any vertex $c \in N_b$, this means that all bad triangles $t$ in this form are an element in  $Y_v$ by \cref{def:triangle-types}. Note that, by \cref{obs:triangle-types} we charge any triangle in 
            $Y_v$ at most by $1+\frac{\epsilon}{1 - \epsilon}$. Confirm that, if we fix $a,b$ and the pivot, there are only $|N_b|$  choices for the third vertex of the bad triangles charged in $y_{(a,b),N_b}$ and each triangle is charged by at most $1 + \frac{\epsilon}{1 - \epsilon}$ as mentioned. Therefore, the total charge of such triangles is at most  $(1+\frac{\epsilon}{1 - \epsilon})|N_b|$.
    \end{enumerate}
    This wraps up the proof of \cref{lem:edge-bound}.
\end{proof}

\begin{claim}\label{clm:expand}
    For any $(a,b) \in E$, it holds that
    \begin{flalign*}
        \E[y_{(a, b)}] &= \Pr[v = a] \cdot \E[y_{(a, b)} \mid v = a] \\ & + \Pr[v = b] \cdot \E[y_{(a, b)} \mid v = b] \\&+ \Pr[v \in (N(a) \Delta N(b)) \setminus \{a,b\} ] \cdot \E[y_{(a, b)} \mid v \in (N(a) \Delta N(b)) \setminus \{a,b\}]  .
    \end{flalign*}
    where $v$ is the first pivot chosen at some iteration in \cref{alg:charging} that after processing $v$, at least one of $a$ or $b$ is removed. 
\end{claim}

\begin{proof} 
    Let us condition on iteration $i$ of the while loop in \cref{alg:charging} being the first iteration where at least one of $a$ or $b$ gets removed from $V$. Note that conditioned on this event, the pivot $v$ of iteration $i$ must be in set $N(a) \cup N(b)$, and note that $a$ and $b$ themselves are part of this set too since $(a, b) \in E$. Moreover, $v$ is chosen uniformly from this set.
    
    By \cref{obs:edge-charged-when-deleted}, no triangle involving $(a, b)$ is charged before or after iteration $i$.  Thus, it suffices to calculate the expected charge to the triangles of $(a, b)$ exactly in iteration $i$. For the rest of the proof, we use $N(u)$ to denote the neighbors of any vertex $u$ still in $V$ in iteration $i$.
    Let us expand $\E[y_{(a, b)}]$ based on whether the pivot $v$ of iteration $i$ is chosen from the common neighbors of $a$ and $b$ or not. We have:
    \begin{flalign*}
        \E[y_{(a, b)}] &= \Pr[v \in N(a) \Delta N(b)] \cdot \E[y_{(a, b)} \mid v \in N(a) \Delta N(b)]\\
        &+ \Pr[v \in N(a) \cap N(b)] \cdot \E[y_{(a, b)} \mid v \in N(a) \cap N(b)].
    \end{flalign*}
    First, by \cref{lem:edge-bound} we have $\E[y_{(a, b)} \mid v \in N(a) \cap N(b)] = 0$. From this, we get that:
    \begin{flalign*}
        \E[y_{(a, b)}] = \Pr[v \in N(a) \Delta N(b)] \cdot \E[y_{(a, b)} \mid v \in N(a) \Delta N(b)].
    \end{flalign*} 
    Note that the structure of our analysis varies when pivot $v$ is chosen as vertex $a$, $b$, or from the set of $(N(a) \Delta N(b)) \setminus \{a,b\}$. To understand the differences we further expand $\E[y_{(a, b)}]$  conditioning each event describing whether $a$, $b$, or a vertex from the union of their neighborhood is chosen as a pivot.
    \begin{flalign*}
        \E[y_{(a, b)}] &= \Pr[v = a] \cdot \E[y_{(a, b)} \mid v = a] \\ & + \Pr[v = b] \cdot \E[y_{(a, b)} \mid v = b] \\&+ \Pr[v \in (N(a) \Delta N(b)) \setminus \{a,b\} ] \cdot \E[y_{(a, b)} \mid v \in (N(a) \Delta N(b)) \setminus \{a,b\}]  .
    \end{flalign*}
\end{proof}

\begin{claim}\label{clm:normal-upper}
    For any $e=(a,b) \in E$ the expected charge on $e$ is at most $$ \frac{(3+\frac{5\epsilon}{1-\epsilon}) \left(|N_a| + |N_b|\right)}{|N_a| + |N_b| + |N_{a,b}| +2}. $$
\end{claim}

\begin{proof}
    By \cref{clm:expand}, we expand  $\E[y_{(a, b)} \mid v \in N(a) \Delta N(b)]$ as follows:
    \begin{flalign*}
        \E[y_{(a, b)}] &= \Pr[v = a] \cdot \E[y_{(a, b)} \mid v = a] \\ & + \Pr[v = b] \cdot \E[y_{(a, b)} \mid v = b] \\&+ \Pr[v \in (N(a) \Delta N(b)) \setminus \{a,b\} ] \cdot \E[y_{(a, b)} \mid v \in (N(a) \Delta N(b)) \setminus \{a,b\}]  .
    \end{flalign*}

    Here we proceed with exploring each possible event for the pivot using \cref{lem:edge-bound}. In the case where $v=a$  for any bad triangle including $a,b$, we charge different values based on the third vertex. Here the charges for each choice of the third vertex $c$ are when $c \in N_a$ and $c \in N_b$: 
            \begin{flalign*}
                \E[y_{(a, b)} \mid v = a] =  & \E[y_{(a, b), N_a} \mid v = a] 
                \\ &+  \E[y_{(a, b), N_b} \mid v = a]  \leq \left(1+\frac{\epsilon}{1 - \epsilon}\right) |N_b| + |N_a|.
            \end{flalign*}

            By rewriting the above inequality for the case where $v =b $ we have:
            \begin{flalign*}
                \E[y_{(a, b)} \mid v = b] =  & \E[y_{(a, b), N_a} \mid v = b] 
                \\ &+  \E[y_{(a, b), N_b} \mid v = b]  \leq \left(1+\frac{\epsilon}{1 - \epsilon}\right) |N_a| + |N_b|.
            \end{flalign*}

            In the last case, where the pivot is not picked as any of $a$ or $b$, we have:
            \begin{flalign*}
                            \E[y_{(a, b)} \mid v \in (N(a) \Delta N(b) ) \setminus \{a,b\}] \leq 1+\frac{4\epsilon}{1 - \epsilon}.
            \end{flalign*}

            Since $\Pr[v = a]  = \Pr[v = b] = \frac{1}{|N(a) \cup N(b)|}$ and $\Pr[v \in (N(a) \Delta N(b)) \setminus \{a,b\} ] = \frac{|N_a|+ |N_b|}{|N(a) \cup N(b)|}$, combining the above inequalities we give the following upper bound for $ \E[y_{(a, b)}]$: 
        \begin{flalign*}
            \E[y_{(a, b)}] \leq & \frac{1}{|N_a| + |N_b| + |N_{a,b}| +2} \left[
            \left(3+\frac{5\epsilon}{1-\epsilon}\right) \left(|N_a| + |N_b|\right)\right]. \qedhere
        \end{flalign*}
\end{proof}

Now, we separate the analysis for three cases, $(C1)-(C3)$, and based on the properties in each case, we determine an upper bound for the expected charge of any edge. We introduce a parameter $\theta$ that will be set to minimize the charge over edges. For any of the following cases, we will use \cref{clm:expand} to expand the expected charge on each edge. To calculate the expected charge of the edge $(a,b)$ conditioned on any event representing the state of the pivot with respect to the pair of $(a,b)$, we need to determine all the bad triangles charged in \cref{alg:charging} in iteration $i$. Note that for the events where $v \in \{a,b\}$, the choices of the third vertex of a bad triangle $t$ in the form of $(a,b,c)$, determines the charges on $t$.

    \begin{enumerate}[label=$(C\arabic*)$]
        \item $\max\{|N_a|, |N_b|\} \leq \frac{\theta}{\delta}$.\label{itm:1-e}
        \item $\max \{|N_a| , N_b| \} > \frac{\theta}{\delta}$, $|N(a) \cap N(b)| + 2< \frac{\delta}{2- \delta} |N(a) \cup N(b)|$.\label{itm:2-e}
        \item $\max \{|N_a| , N_b| \} > \frac{\theta}{\delta}$, $|N(a) \cap N(b)| + 2 \geq \frac{\delta}{2- \delta} |N(a) \cup N(b)|.$\label{itm:3-e}
    \end{enumerate}

\begin{claim}\label{clm:1-e}
    In \ref{itm:1-e}, the expected charge on $(a,b)$ is at most $\left(1- \frac{\delta}{\theta + \delta}\right)\left(3+ \frac{5\epsilon}{1 - \epsilon}\right)$. 
\end{claim}
\begin{proof}

    To prove the claim, we use the upper bound from \cref{clm:normal-upper} and the condition in \ref{itm:1-e}: 
        \begin{flalign*}
            \E[y_{(a, b)}] \leq & \frac{1}{|N_a| + |N_b| + |N_{a,b}| +2} \left[
            \left(3+\frac{5\epsilon}{1-\epsilon}\right) (|N_a| + |N_b|)\right] 
            \\& \leq \left(1 - \frac{2}{|N_a| + |N_b| + 2}\right)\left(3+\frac{5\epsilon}{1 - \epsilon}\right)
            \leq \left(1- \frac{\delta}{\theta + \delta}\right)\left(3+ \frac{5\epsilon}{1 - \epsilon}\right).\qedhere
        \end{flalign*}
    \end{proof}

    \begin{claim}\label{clm:2-e}
    In \ref{itm:2-e}, the expected charge on $(a,b)$ is at most $3 +\frac{5\epsilon}{1 - \epsilon} - \frac{\theta\delta +\delta^2-\delta}{2(\theta+\delta )}\cdot \frac{2-7\delta}{2-3\delta}$. 
\end{claim}
\begin{proof} 

    Let us assume that  $N(b) \leq N(a)$, by this distinction between $a$ and $b$, we investigate each event representing different states for pivot:
    \begin{enumerate}
        
    \item $v =a$:

        In this event, we charge the pair $(a,b)$ for any remaining vertex $c$ in the union of the neighborhood of $a$ and $b$, this is because, any bad triangle has $2$ adjacent vertices, and since we are charging all the bad triangles involving $a,b$, the third vertex should be either adjacent to $a$ or $b$.
        Now, by investigating any choice of vertex $c \in N(a) \cup N(b)$ that creates a bad triangle with $a,b$, we compute the total charges on $a,b$. Note that, the different cases affecting the analysis, are related to whether $c$ is picked from $N_a$, $N_b$, or $N_{a,b}$, we expand $\E[y_{(a, b)} \mid v = a]$ based on these choices for the third vertex:
            \begin{flalign*}
                \E[y_{(a, b)}  \mid v = a ] &=  \E[y_{(a, b),N_a} \mid v = a ] 
                \\ & + \E[y_{(a, b), N_{a,b}} \mid v = a] 
                \\&+ \E[y_{(a, b),N_b} \mid v =a]  
                \\& \leq  \E[y_{(a, b),N_a} \mid v = a ]  + \left(1 + \frac{\epsilon}{1 - \epsilon}\right)|N_b|.
            \end{flalign*}
            
            Note that the inequality is resulted from \cref{lem:edge-bound}. Now we explore $\E[y_{(a, b),N_a} \mid v = a ]$.  

               Since $N(b) \leq N(a)$, and based on the assumption of this claim,  we have
    $$|N(a) \cap N(b)| +2 < \frac{\delta}{2 - 2\delta} (|N_a| + |N_b| ) \leq \frac{\delta}{1- \delta} |N_a|. $$
    Moving the terms, this implies
               $$(1- \delta)(|N(a) \cap N(b)| +2) < \delta (|N_a|),$$
               which using the fact that $N_a = N(a) \setminus (N(b) \cup b)$ it holds that:
               $$|N(a) \cap N(b)| +1 < \delta (|N(a)| + 2) - 1  < \delta|C_v|.$$

                Note that the above inequality implies that $|N(a) \cap N(b)| +1 < \delta |C_v|$ by \cref{alg:modified-pivot}, we have $b \in D_a$. Thus, vertex $b$ joins $D_v'$ with probability $\frac{\min\{|D_v|, \lfloor \delta |C_v| \rfloor \} } { |D_v|}$. Here we find a lower bound for this probability using the condition in \ref{itm:2-e}:
                \begin{flalign*}
                    \frac{\min\{|D_v|, \lfloor \delta |C_v| \rfloor \} } { |D_v|} 
                    \geq  \frac{\delta|C_v| - 1}{|D_v|} 
                     \geq \delta - \frac{\delta}{\theta + \delta}
                \end{flalign*}

                Note that by \cref{obs:edge-charged-when-deleted} any edge is charged once, and then at least one of its endpoints is removed from the graph.
                 The only choices of $c$ that change the charging of $t$ depending on whether $D_a'$ contains $b$ or not, are the vertices in $N_a$. At this step, we can expand $\E[y_{(a, b),N_a} \mid v = a ] $ conditioning on state of $b$ with respect to $D_v'$:\begin{flalign*}
                     \E[y_{(a, b),N_a} \mid v = a ] &= \Pr[b \notin D_v'|v=a] \cdot \E[y_{(a, b), N_a} \mid v = a, b \notin D_v'] \\
                    &+ \Pr[b \in D_v'|v=a ] \cdot \E[y_{(a, b),N_a} \mid v = a, b \in D_v'].
                \end{flalign*}
                 
                 In the first case, if $b \notin D_a'$: if $c \notin D_a'$ we charge $t$ by \cref{line:C1}, otherwise we charge it by \cref{line:C2}. Therefore in this case for each choice of $c$, we charge $t$ at most 1, and since we have $|N_a|$ such bad triangles then:
                 \begin{flalign*}
                    \E[y_{(a, b),N_a} \mid v = a, b \notin D_v'] \leq |N_a|.
                \end{flalign*}
                
                In the case where $b \in D_a'$ we always charge $t$ by \cref{line:C2}. This implies the following:
                \begin{flalign*}
                    \E[y_{(a, b),N_a} \mid v = a, b \in D_v'] = \frac{2\delta}{1-\frac{3}{2}\delta}|N_a| .
                \end{flalign*}
                
                Based on the bounds above, we get: 
                \begin{flalign*}
                     \E[y_{(a, b),N_a} \mid v = a ] \leq
                     &  \left( \left(1-\frac{\min\{|D_v|, \lfloor \delta |C_v| \rfloor \} } { |D_v|} \right) 
                    + \frac{\min\{|D_v|, \lfloor \delta |C_v| \rfloor \} } { |D_v|} \cdot \frac{2\delta}{1-\frac{3}{2}\delta} \right) |N_a|
                    \\& \leq \left(1- \frac{\min\{|D_v|, \lfloor \delta |C_v| \rfloor \} } { |D_v|} \left(1 -  \frac{2\delta}{1-\frac{3}{2}\delta} \right)\right) |N_a|
                    \\& \leq \left(1-  \frac{\theta\delta +\delta^2 -\delta}{\theta +\delta}\cdot \frac{2-7\delta}{2-3\delta}\right)|N_a|
                \end{flalign*}

            \item $v=b$: 

            As explored in event $v=a$, we differentiate between triangles by choices of the third vertex in $t$. Following this we expand $ \E[y_{(a, b)}  \mid v = b ]$:
            \begin{flalign*}
            \E[y_{(a, b)}  \mid v = b ] &=  \E[y_{(a, b),N_b} \mid v = b ] 
             + \E[y_{(a, b), N_{a,b}} \mid v = b] 
            + \E[y_{(a, b),N_a} \mid v =b]
            \\ &\leq \left(1+\frac{\epsilon}{1 - \epsilon}\right)|N_a| +  |N_b|.
            \end{flalign*}

            Confirm that the above inequality is simply resulted from \cref{lem:edge-bound}.

            \item  $v \in (N(a) \Delta N(b) ) \setminus \{a,b\}$: 
            
                Directly by \cref{lem:edge-bound} we have:
                \begin{flalign*}
                \E[y_{(a, b)} \mid v \in (N(a) \Delta N(b) ) \setminus \{a,b\}] \leq \left(1 + \frac{4\epsilon}{1 - \epsilon}\right).
            \end{flalign*}

        \end{enumerate}
        
         Finally, we have:
         \begin{flalign*}
        \E[y_{(a, b)}] &= \Pr[v = a] \cdot\left( \left(1-  \frac{\theta\delta +\delta^2 -\delta}{\theta +\delta}\cdot \frac{2-7\delta}{2-3\delta}\right) |N_a| + \left(1+\frac{\epsilon}{1 - \epsilon}\right) |N_b|\right)
        \\ & + \Pr[v = b] \cdot \left[\left(1+ \frac{\epsilon}{1 - \epsilon}\right)|N_a| + |N_b| \right]
        \\& + \Pr[v \in (N(a) \Delta N(b)) \setminus \{a,b\} ] \cdot \left(1+\frac{4\epsilon}{1 - \epsilon}\right)
        \\ &=\frac{\left(3 + \frac{5\epsilon}{1 - \epsilon} - \frac{\theta\delta + \delta ^2 -\delta}{\theta+ \delta}\cdot \frac{2-7\delta}{2-3\delta}\right) |N_a| + \left(3+\frac{5\epsilon}{1 - \epsilon}\right) |N_b|}{|N_a| + |N_b| + |N_{a,b}| + 2 }.
        \end{flalign*}

        Let $\alpha = \frac{ \frac{\theta\delta + \delta ^2 -\delta}{\theta+ \delta}\cdot \frac{2-7\delta}{2-3\delta}}{3+ \frac{5\epsilon}{1 - \epsilon}}$. Now, we give an upper bound on $E[y_{(a, b)}]$ based on $\alpha$:
        \begin{flalign*}
        E[y_{(a, b)}] &\leq \frac{3+\frac{5\epsilon}{1 - \epsilon}}{|N_a| + |N_b| + |N_{a,b}| + 2}
        \left[ (1-\alpha) |N_a|  + (1-\frac{\alpha}{2})|N_b| + \frac{\alpha}{2}|N_b| \right] 
        \\ & \leq  \frac{3+\frac{5\epsilon}{1 - \epsilon}}{|N_a| + |N_b| + |N_{a,b}| + 2}
        \left[ (1-\frac{\alpha}{2}) |N_a|  + (1-\frac{\alpha}{2})|N_b| \right]
        \\ &\leq \left(3+\frac{5\epsilon}{1 - \epsilon}\right)\left(1-\frac{\alpha}{2}\right)
        \\&  = 3 +\frac{5\epsilon}{1 - \epsilon} - \frac{\theta\delta +\delta^2-\delta}{2(\theta+\delta )}\cdot \frac{2-7\delta}{2-3\delta}.\qedhere
        \end{flalign*}
        
\end{proof}
 \begin{claim}\label{clm:3-e}
    In \ref{itm:3-e}, the expected charge on $(a,b)$ is at most $\left(1- \frac{\delta}{2-\delta}\right)\left(3+\frac{5\epsilon}{1 - \epsilon}\right)$. 
\end{claim}
   \begin{proof}
       Note that by the condition in \ref{itm:3-e}, we have: $$|N(a) \cap N(b)| + 2 \geq \frac{\delta}{2- \delta} |N(a) \cup N(b)|,$$ this implies that:

       $$|N_a| + |N_b| \leq \left( 1- \frac{\delta}{2-\delta}\right) |N(a) \cup N(b)|.$$ 
       Using the inequality on the sum of $|N_a|$ and $|N_b|$, and also the upper bound from \cref{clm:normal-upper} we have:
        \begin{flalign*}
            \E[y_{(a, b)}] 
            \leq \frac{1}{|N_a| + |N_b| + |N_{a,b}| +2} \left[
            \left(3+\frac{5\epsilon}{1 - \epsilon}\right) (|N_a| + |N_b|)\right]  \leq \left(1- \frac{\delta}{2-\delta}\right)\left(3+\frac{5\epsilon}{1 - \epsilon}\right).\qedhere
        \end{flalign*}
\end{proof}

% \begin{proof}[Proof of \cref{clm:pair-bound} for $(a, b) \in E$]

\subsubsection{Width Analysis for Non-edges}\label{sec:width-ne}

\begin{claim}\label{lem:non-edge-bound}
    For any $(a,b) \notin E$ we have:
    \begin{enumerate}
        \item $\E[y_{(a, b), N_a\cup N_b }] =  0$ .
        \item $\E[y_{(a, b)} \mid v \in N_{a,b} ] \leq 1$ .
        \item $\E[y_{(a, b), N_{a,b}} \mid v =a ] \leq (1+\frac{\epsilon}{1 - \epsilon})|N_{a,b}|$  .
        
    \end{enumerate}
\end{claim}
\begin{proof}
    We prove the three parts one by one.
    \begin{enumerate}
        \item Note that, the triangle $t =(a,b,c)$ such that $c \in N_a\cup N_b$ does not form a bad triangle as there exists only one edge in $t$.
        
        \item In this case, we have $v \in N_{a,b}$ that means the pivot $v$ is adjacent to both $a$ and $b$. By \cref{def:triangle-types} any bad triangle of this structure belongs to the set $X_v$. By \cref{obs:triangle-types} we charge such bad triangles at most by $1$. Note that, for any fixed pair of vertices given the pivot, we have one such bad triangle, and therefore the total charge is bounded by $1$.
        
        \item Note that any triangle charged in this case is in $Y_v$. This is because for any fixed pair of non-edge $(a,b)$, any bad triangle charged in $y_{(a,b),N_{a,b}}$ with the condition that $v=a$, we have $v$ is not adjacent to $b$ but it is adjacent to the third vertex $c$ chosen from the set $N_{a,b}$. By \cref{def:triangle-types} any such bad triangle is in $Y_v$ and is charged at most by $1 + \frac{\epsilon}{1 - \epsilon}$ as we discussed in \cref{obs:triangle-types}.  Summing up over choices of the third vertex, we get an upper bound of  $(1+\frac{\epsilon}{1 - \epsilon})|N_{a,b}|$ over charges to all such bad triangles.
    \end{enumerate}
    The proof is complete.
\end{proof}

\begin{claim}\label{clm:expand-n}
    The expected charge over a pair of vertices $(a,b) \notin E$ is expandable as follows in case the pair does not belong to the set $E$:
     \begin{flalign*}
        \E[y_{(a, b)}] &= \Pr[v = a] \cdot \E[y_{(a, b)} \mid v = a] \\ & + \Pr[v = b] \cdot \E[y_{(a, b)} \mid v = b] \\&+ \Pr[v \in N_{a,b} ] \cdot \E[y_{(a, b)} \mid v \in N_{a,b} ]  + \frac{5}{k} \cdot \frac{\epsilon}{1-\epsilon}.
    \end{flalign*}
    where $v$ is the first pivot chosen at some iteration in \cref{alg:charging} that after processing $v$, at least one of $a$ or $b$ is removed. 
\end{claim}

\begin{proof} 
    Let us condition on iteration $i$ of the while loop in \cref{alg:charging} being the first iteration where at least one of $a$ or $b$ gets removed from $V$. Note that conditioned on this event, the pivot $v$ of iteration $i$ must be in set $N(a) \cup N(b) \cup \{a,b\}$. Moreover, $v$ is chosen uniformly from this set.

    By \cref{obs:edge-charged-when-deleted}, any triangle involving $(a, b)$ is charged in at most two iterations.  We consider the charge from the iteration that results in removing at least one of the endpoints of this pair (iteration $i$), and sum it up with the maximum possible charge that could have happened in \cref{line:C7} of an earlier iteration in \cref{alg:charging}.   For the rest of the proof, we use $N(u)$ to denote the neighbors of any vertex $u$ still in $V$ in iteration $i$.
    
    Let us expand $\E[y_{(a, b)}]$ based on whether the pivot $v$ of iteration $i$ is chosen from the common neighbors of $a$ and $b$ or not. We use $\E[y_{(a,b)} \mid v']$ to denote the additive expected charge for $(a,b)$ resulted from the case where $(a,b)$ is charged once before iteration $i$, and we use $v'$ to denote the pivot picked at that earlier iteration. Taking this charge into account, it holds that: 
    \begin{flalign*}
        \E[y_{(a, b)}] &= \Pr[v \in N(a) \Delta N(b)] \cdot \E[y_{(a, b)} \mid v \in N(a) \Delta N(b)]\\
        &+ \Pr[v \in N_{a,b} \cup \{a,b\}] \cdot \E[y_{(a, b)} \mid v \in N_{a,b} \cup \{a,b\}] + \E[y_{(a,b)} \mid v'] 
    \end{flalign*}

    First, note that by \cref{lem:non-edge-bound}, $\E[y_{(a, b)} \mid v \in N(a) \Delta N(b)] = 0$. From this, we get that
    \begin{flalign*}
        \E[y_{(a, b)}] = \Pr[v \in N_{a,b} \cup \{a,b\}] \cdot \E[y_{(a, b)} \mid v \in N_{a,b} \cup \{a,b\}] + \E[y_{(a,b)} \mid v'].
    \end{flalign*} 
    Note that the structure of our analysis varies when pivot $v$ is chosen as vertex $a$, $b$, or from the set of $N_{a,b} $. To understand the differences we further expand $\E[y_{(a, b)}]$  conditioning on each event describing whether $a$, $b$, or a vertex from the intersection of their neighborhood is chosen as a pivot.
    \begin{flalign*}
        \E[y_{(a, b)}] &= \Pr[v = a] \cdot \E[y_{(a, b)} \mid v = a] \\ & + \Pr[v = b] \cdot \E[y_{(a, b)} \mid v = b] \\&+ \Pr[v \in N_{a,b} ] \cdot \E[y_{(a, b)} \mid v \in N_{a,b} ] + \E[y_{(a,b)} \mid v']  .
    \end{flalign*}

    Now, it only remains to prove that $\E[y_{(a,b)} \mid v'] \leq \frac{5}{k} \cdot \frac{\epsilon}{1-\epsilon} $. Note that, there exists at most one pivot $v'$ charging any non-edge by \cref{line:C7} in \cref{alg:charging}, however, for any third vertex $c$ holding the properties of vertex $u$ in \cref{line:C7} in iteration where we remove $v'$, we charge $(a,b,c)$ by $\frac{\frac{5\epsilon}{1-\epsilon}}{|A_{v'}|-1}$. Since there are most $|C_{v'}|$ choices for $c$, this gives an upper bound of $\frac{\frac{5\epsilon}{1-\epsilon}}{|A_{v'}|-1} \cdot |C_{v'}|$ for this particular charges on $(a,b)$. Since we only charge such bad triangles if $|A_{v'}| > k |C_{v'}|$, this implies 
    $$\E[y_{(a,b)} \mid v'] \leq \frac{\frac{5\epsilon}{1-\epsilon}}{k}. \qedhere$$
\end{proof}

\begin{claim}\label{clm:non-normal-upper}
    For any $e=(a,b) \notin E$  the expected charges over $e$ is at most $$ \frac{(3+\frac{2\epsilon}{1 - \epsilon}) |N_{a,b}|}{|N_a| + |N_b| + |N_{a,b}| +2} + \frac{\frac{5\epsilon}{1-\epsilon}}{k}
              .$$
\end{claim}

\begin{proof}
    Note that by \cref{clm:expand-n} we have:
    \begin{flalign*}
        \E[y_{(a, b)}] &= \Pr[v = a] \cdot \E[y_{(a, b)} \mid v = a] \\ & + \Pr[v = b] \cdot \E[y_{(a, b)} \mid v = b] \\&+ \Pr[v \in N_{a,b} ] \cdot \E[y_{(a, b)} \mid v \in  ]  + \frac{5}{k} \cdot \frac{\epsilon}{1-\epsilon}.
    \end{flalign*}

    Here we proceed with exploring each event using \cref{lem:non-edge-bound}. In the case where $v=a$ for any bad triangle including $a,b$, we charge different values based on the third vertex. Here the charges for each choice of the third vertex $c$ are when $c \in N_{a,b}$: 
            \begin{flalign*}
                &\E[y_{(a, b)} \mid v \in \{a,b\}] \\ &=   \E[y_{(a, b), N_{a,b}} \mid v = a] + \E[y_{(a, b), N_{a,b}} \mid v = b]
                \\  &\leq 2(1+\frac{\epsilon}{1 - \epsilon}) |N_{a,b}|.
            \end{flalign*}
For the case that the pivot is picked from the common neighbors of $a$ and $b$, we get:
            \begin{flalign*}
                            \E[y_{(a, b)} \mid v \in N_{a,b}] \leq 1.
            \end{flalign*}
        
        Since $\Pr[v = a]  = \Pr[v = b] = \frac{1}{|N(a) \cup N(b) \cup \{a,b\}|}$ and $\Pr[v \in N_{a,b} ]  = \frac{|N_{a,b}|}{|N(a) \cup N(b) \cup \{a,b\}|}$ , 
        combining the above inequalities we give the following upper bound for $ \E[y_{(a, b)}]$: 
        \begin{flalign*}
            \E[y_{(a, b)}] & \leq \frac{1}{|N_a| + |N_b| + |N_{a,b}| +2} \left[
            \left(3+\frac{2\epsilon}{1 - \epsilon}\right) |N_{a,b}| \right]  + \frac{\frac{5\epsilon}{1-\epsilon}}{k}.\qedhere
            \end{flalign*}

\end{proof}

Now, we separate the analysis for three cases, $(D1)-(D3)$, and based on the properties in each case, we determine an upper bound for the expected charge of any edge. We introduce a parameter $\lambda$ that will be set to minimize the charge over non-edges. For any of the following cases, we will use \cref{clm:expand-n} to expand the expected charge on each edge. To calculate the expected charge of the non-edge $(a,b)$ conditioned on any event representing the state of the pivot with respect to the pair of $(a,b)$, we need to determine all the bad triangles charged in \cref{alg:charging} in iteration $i$. Note that for the events where $v \in \{a,b\}$, the choices of the third vertex of a bad triangle $t$ in the form of $(a,b,c)$, determines the charges on $t$.

    \begin{enumerate}[label=$(D\arabic*)$]
        \item $\min\{|N(a)|, |N(b)| \}  \leq \frac{\lambda}{\delta}$.\label{itm:1-ne}
        \item $\min\{|N(a)|, |N(b)| \}  > \frac{\lambda}{\delta}$, $|N(a) \Delta N(b)|+2 < \frac{\epsilon}{1+\epsilon} |N(a) \cup N(b)|$.\label{itm:2-ne} 
        \item $\min\{|N(a)|, |N(b)| \}  > \frac{\lambda}{\delta}$, $|N(a) \Delta N(b)|+2 \geq \frac{\epsilon}{1+\epsilon} |N(a) \cup N(b)|.$\label{itm:3-ne}
    \end{enumerate}

\begin{claim}\label{clm:1-ne}
     Let us assume that $|N_a| \geq |N_b|$ w.l.o.g. In \ref{itm:1-ne}, the expected charge on $(a,b)$ is at most $\frac{\lambda+\delta}{(2\delta+\lambda)}\left(3+\frac{2\epsilon}{1 - \epsilon}\right) + \frac{\frac{5\epsilon}{1-\epsilon}}{k}.$ 
\end{claim}
\begin{proof}
By \cref{clm:non-normal-upper} and the condition in \ref{itm:1-ne} we have:
        \begin{flalign*}
            \E[y_{(a, b)}] &\leq \frac{1}{|N_a| + |N_b| + |N_{a,b}| +2} \left[
            \left(3+\frac{2\epsilon}{1 - \epsilon}\right) |N_{a,b}| \right] + \frac{\frac{5\epsilon}{1-\epsilon}}{k} 
            \\& \leq \left(\frac{|N(b)|}{|N(b)| + 2 }\right)\left(3+\frac{2\epsilon}{1 - \epsilon}\right) + \frac{\frac{5\epsilon}{1-\epsilon}}{k} 
            \\ & \leq   \left(1 - \frac{1}{|(N(b)| + 2}\right)\left(3+\frac{2\epsilon}{1 - \epsilon}\right) +\frac{\frac{5\epsilon}{1-\epsilon}}{k}
            \\ & \leq \frac{\lambda + \delta}{(2\delta+\lambda)}\left(3+\frac{2\epsilon}{1 - \epsilon}\right) +\frac{\frac{5\epsilon}{1-\epsilon}}{k}.\qedhere
        \end{flalign*}

        Note that the second inequality holds since we have $|N(b)| \geq |N_{a,b}|$ and the claim assumption implies $|N(b)|+2 \leq |N_a| + |N_b| + |N_{a,b}| +2 $. Also, the last inequality holds since we have $|N(b)| \leq \frac{\lambda}{ \delta}$, this concludes that $1 - \frac{1}{|N(b)| + 2} \leq 1 - \frac{1}{\lambda/\delta +2} = \frac{\lambda + \delta}{\lambda + 2\delta}$. 
        
\end{proof}

\begin{claim}\label{clm:2-ne}
    In \ref{itm:2-ne}, the expected charge on $(a,b)$ is at most $$ \max  \left[ 3 +\frac{2\epsilon}{1 - \epsilon}+ 2\left( - \frac{\delta}{k} + \frac{\delta/k}{\delta + \lambda}\right)\cdot \left(1+\frac{\epsilon}{1 - \epsilon} - \delta\right) , 3- \frac{2\epsilon}{1 - \epsilon} \right] + \frac{\frac{5\epsilon}{1-\epsilon}}{k}.$$
\end{claim}
\begin{proof}
        Here the analysis varies when pivot $v$ is chosen as vertex $a$, $b$, or from the set of $N_{a,b}$. To understand the differences we further expand $\E[y_{(a, b)}]$ by \cref{clm:expand-n} conditioning on whether $a$ or $b$ is chosen as a pivot or not:
            \begin{flalign*}
                \E[y_{(a, b)}] &\leq \Pr[v = a] \cdot \E[y_{(a, b)} \mid v = a] 
                \\ & + \Pr[v = b] \cdot \E[y_{(a, b)} \mid v = b] 
                \\&+ \Pr[v \in  N_{a,b} ] \cdot \E[y_{(a, b)} \mid v \in N_{a,b}] +\frac{\frac{5\epsilon}{1-\epsilon}}{k}.
            \end{flalign*}
        
             We determine all the bad triangles charged in \cref{alg:charging} in iteration $i$ by investigating each event based on the pivot separately:
        \begin{enumerate}
        
            \item $v \in \{a,b\}$:

                Now, by checking any vertex $c \in N_{a,b}$, we find about each charging in \cref{alg:charging} that charges triangle $t = (a,b,c)$. We explore $\E[y_{(a, b)} \mid v = a ]$, and note that the analysis for the case where $v=b$ is the same as that for $v=a$.  Now, the condition in \ref{itm:2-ne} implies 
                $$(1+ \epsilon) (|N_a| + |N_b|) +2 < \epsilon |N(a) \cup N(b)|,$$
                 which in turn, results in 
                 $$|N_a| + |N_b| +2 < \epsilon (|N_{a,b}| +2) < \epsilon |N_{a,b}|+1 .$$
                 Note that we have $$N(a) \Delta N(b) = N_a \cup N_b.$$
                 This implies that:
               \begin{align*}
                    |N(a) \Delta N(b)| +1 \leq \epsilon|N_{a,b}| \leq \epsilon |N(a)|.
                \end{align*}

                Note that the above inequality implies that $|N(a) \Delta N(b)| < \epsilon(|N(a)| + 1) -1 = \epsilon |C_v| -1$ and therefore we can conclude $b \in A_v$. Observe that by \cref{alg:charging}, the vertex $b$ joins $A_v'$ with probability $\frac{\min\{|A_v|, \lfloor \delta |C_v| \rfloor\} } { |A_v|}$.Note that $t \in Y_i$, and therefore the charges on different triangles vary whether of $b \in A_v'$ or not.  We also have two different charging schemes based on the size of $A_v$.

                \begin{itemize}
                    
                    \item $|A_v| \leq k|C_v|$:
                        In this case, by the condition in \cref{clm:2-ne}, we have $-\frac{1}{k|C_v|} \geq -\frac{1/k}{1 + \lambda/\delta}$. Thus we have:
                        \begin{flalign*}
                    &\frac{\min\{|A_v|, \lfloor \delta |C_v| \rfloor\} } { |A_v|} \geq \frac{\delta |C_v|-1}{k|C_v|} \geq \frac{\delta}{k}-\frac{\delta/k}{\delta + \lambda}.
                \end{flalign*}
                        When the size of $A_v$ is not too large compared to that of $C_v$, we charge any triangle $t$ by $\delta$ if $b \in A_v'$ and $1+\frac{\epsilon}{1 - \epsilon}$ otherwise. Based on the probability that $b$ is chosen as a member of $A_v'$, the expected number of triangles charged containing $(a,b)$ can be written as follows:
                        \begin{flalign*}
                     \E[y_{(a, b),N_{a,b}} \mid v = a ] = & 
                      \Pr[b \notin A_v' | v=a] \cdot \E[y_{(a, b), N_{a,b}} \mid v = a, b \notin A_v'] 
                    \\+& \Pr[b \in A_v' |v=a] \cdot \E[y_{(a, b),N_{a,b}} \mid v = a, b \in A_v'].
                \end{flalign*}

                    In the first case, if $b \notin A_v'$ we charge $t$ by \cref{line:C4}, Therefore in this case for each choice of $c$, we charge $t$ at most $1+\frac{\epsilon}{1 - \epsilon}$, precisely we have:
                 \begin{flalign*}
                    \E[y_{(a, b),N_{a,b}} \mid v = a, b \notin A_v'] = \left(1+\frac{\epsilon}{1 - \epsilon}\right)|N_{a,b}|.
                \end{flalign*}

                 In the case where $b \in A_v'$ we always charge $t$ by \cref{line:C3}. This implies the following:
                \begin{flalign*}
                    \E[y_{(a, b),N_{a,b}} \mid v = a, b \in A_v'] = \delta |N_{a,b}|.
                \end{flalign*}

                Using the expected charges above the following equality holds: 
                     \begin{flalign*}
                     \E[y_{(a, b),N_{a,b}} \mid v = a ] =
                     &  \left(1-\frac{\min\{|A_v|, \lfloor \delta |C_v| \rfloor \} } { |A_v|}\right) \left(1+\frac{\epsilon}{1 - \epsilon}\right)|N_{a,b}|
                    \\&+\left( \frac{\min\{|A_v|, \lfloor \delta |C_v| \rfloor \} } { |A_v|} \cdot \delta \right) |N_{a,b}|
                    \\&= \left(1+\frac{\epsilon}{1 - \epsilon} - \frac{\min\{|A_v|, \lfloor \delta |C_v| \rfloor \}\left(1+\frac{\epsilon}{1 - \epsilon} - \delta\right) } { |A_v|} \right)|N_{a,b}| 
                    \\& \leq \left(1+\frac{\epsilon}{1 - \epsilon}+ \left( - \frac{\delta}{k} + \frac{\delta/k}{\delta + \lambda}\right)\cdot \left(1+\frac{\epsilon}{1 - \epsilon} - \delta\right)\right) |N_{a,b}|.
                \end{flalign*}

            \item $ |A_v| > k|C_v|$:
                    When the size of $A_v$ is significantly larger than that of $C_v$, we always charge triangle $t$ by $1 - \frac{\epsilon}{1-\epsilon}$ in \cref{line:C6}: 
                \begin{flalign*}
                     \E[y_{(a, b),N_{a,b}} \mid v = a ] = \left(1-\frac{\epsilon}{1 - \epsilon}\right) |N_{a,b}|.
                \end{flalign*}

                \end{itemize}

            \item $v \in N_{a,b}:$
                Directly by \cref{lem:non-edge-bound} we have:
                \begin{flalign*}
                \E[y_{(a, b)} \mid v \in N_{a,b}] \leq 1 .
            \end{flalign*}
            \end{enumerate}

            Finally, we can give an upper bound for the expected charges on $(a,b)$ by the maximum charge in the above cases:
            \begin{flalign*} 
        \nonumber \E[y_{(a, b)}] &\leq  \Pr[v \in \{a,b\}] \cdot \max  \left[\left(1+\frac{\epsilon}{1 - \epsilon}+ \left( - \frac{\delta}{k} + \frac{\delta/k}{\delta + \lambda}\right)\cdot \left(1+\frac{\epsilon}{1 - \epsilon} - \delta\right)\right) , 1 - \frac{\epsilon}{1 - \epsilon}\right] |N_{a,b}|  
        \\
        \nonumber&+ \Pr[v \in N_{a,b} ] + \frac{\frac{5\epsilon}{1-\epsilon}}{k}
        \\
        \nonumber& = \frac{ \max  \left[ 3 +\frac{2\epsilon}{1 - \epsilon}+ 2\left( - \frac{\delta}{k} + \frac{\delta/k}{\delta + \lambda}\right)\cdot \left(1+\frac{\epsilon}{1 - \epsilon} - \delta\right) , 3- \frac{2\epsilon}{1 - \epsilon} \right]  }{|N_a| + |N_b| + |N_{a,b}| + 2 }|N_{a,b}|  + \frac{\frac{5\epsilon}{1-\epsilon}}{k}
        \\
        &\leq  \max  \left[ 3 +\frac{2\epsilon}{1 - \epsilon}+ 2\left( - \frac{\delta}{k} + \frac{\delta/k}{\delta + \lambda}\right)\cdot \left(1+\frac{\epsilon}{1 - \epsilon} - \delta\right) , 3- \frac{2\epsilon}{1 - \epsilon} \right] + \frac{\frac{5\epsilon}{1-\epsilon}}{k}. \qedhere
        \end{flalign*}

\end{proof}
\begin{claim}\label{clm:3-ne}
    In \ref{itm:3-ne}, the expected charge on $(a,b)$ is at most $\left(1- \frac{\epsilon}{1+\epsilon}\right)\left(3+\frac{2\epsilon}{1 - \epsilon}\right) + \frac{\frac{5\epsilon}{1-\epsilon}}{k}.$ 
\end{claim}   
\begin{proof}

    By \cref{clm:non-normal-upper} and the condition in \ref{itm:3-ne} we have:
        \begin{flalign*}
            \E[y_{(a, b)}] = \frac{1}{|N_a| + |N_b| + |N_{a,b}| +2} \left[
            \left(3+\frac{2\epsilon}{1 - \epsilon}\right) |N_{a,b}| \right]  \leq \left(1- \frac{\epsilon}{1+\epsilon}\right)\left(3+\frac{2\epsilon}{1 - \epsilon}\right) + \frac{\frac{5\epsilon}{1-\epsilon}}{k}.\qedhere
        \end{flalign*}

\end{proof}

Finally, we are ready to wrap up the proof of \cref{clm:pair-bound}:
\begin{proof}[Proof of \cref{clm:pair-bound} for any pair $(a, b)$]
Now, looking through the width analysis for edges and non-edges, to prove \cref{clm:pair-bound}, for any case described in \cref{sec:width-e} and \cref{sec:width-ne}, we introduce a set of values for parameters $\epsilon, \delta, \lambda, $ and $\theta$ that imply a \apxfactor-approximation.
 We set $\epsilon =  0.007$,
$\delta = 0.179$,
$\lambda = 7.613$,
$ \theta = 7.055$, and
$ k = 12.295$.

For any edge in $E$, we investigate the three cases $(C1)-(C3)$. For each case, we prove that $E[Y_{a,b}] < \apxfactor$.  

\begin{itemize}
    \item In \ref{itm:1-e}, by the upper bound in \cref{clm:1-e} and plugging in the parameters with introduced values we get:
    $$E[Y_{a,b}]  \leq \left(1- \frac{\delta}{\theta + \delta}\right)\left(3+ \frac{5\epsilon}{1 - \epsilon}\right)  < 2.961.$$

    \item In \ref{itm:2-e}, by the upper bound in \cref{clm:2-e} and plugging in the parameters with introduced values we get: 
    $$E[Y_{a,b}]  \leq 3 +\frac{5\epsilon}{1 - \epsilon} - \frac{\theta\delta +\delta^2-\delta}{2(\theta+\delta )}\cdot \frac{2-7\delta}{2-3\delta} < 2.996.$$

    \item In \ref{itm:3-e}, by the upper bound in \cref{clm:3-e} and plugging in the parameters with introduced values we get: 
    $$E[Y_{a,b}]  \leq 3 +\left(1- \frac{\delta}{2-\delta}\right)\left(3+\frac{5\epsilon}{1 - \epsilon}\right) < 2.737.$$
    
\end{itemize}

 For any non-edge in $E$, we investigate the three cases $(D1)-(D3)$. For each case, we prove that $E[Y_{a,b}] < \apxfactor$.  

 \begin{itemize}

    \item In \ref{itm:1-ne}, by the upper bound in \cref{clm:1-ne} and plugging in the parameters with introduced values we get: 
    $$E[Y_{a,b}]  \leq \frac{\lambda+\delta}{2\delta+\lambda}\left(3+\frac{2\epsilon}{1 - \epsilon}\right) + \frac{\frac{5\epsilon}{1-\epsilon}}{k} <2.95.$$
    
     \item In \ref{itm:2-ne}, by the upper bound in \cref{clm:2-ne} and plugging in the parameters with introduced values we get: 
    $$E[Y_{a,b}]  \leq \max  \left[ 3 +\frac{2\epsilon}{1 - \epsilon}+ 2\left( - \frac{\delta}{k} + \frac{\delta/k}{\delta + \lambda}\right)\cdot \left(1+\frac{\epsilon}{1 - \epsilon} - \delta\right) , 3- \frac{2\epsilon}{1 - \epsilon} \right] + \frac{\frac{5\epsilon}{1-\epsilon}}{k} <  2.996.$$

     \item In \ref{itm:3-ne}, by the upper bound in \cref{clm:3-ne} and plugging in the parameters with introduced values we get: 
    $$E[Y_{a,b}]  \leq \left(1- \frac{\epsilon}{1+\epsilon}\right) \left(3+\frac{2\epsilon}{1 - \epsilon}\right) + \frac{\frac{5\epsilon}{1-\epsilon}}{k} < 2.997.$$
 \end{itemize}

This concludes the proof of \cref{clm:pair-bound}.
\end{proof}

%% file: sublinear.tex
\section{Implementation in the Fully Dynamic Model}

In this section, we prove \cref{thm:dynamic} that a $(3-\Omega(1))$-approximation of correlation clustering can be maintained by spending polylogarithmic time per update.

\begin{proof}[Proof of \cref{thm:dynamic}]
    Our starting point is the algorithm of \citet*{BehnezhadDHSS-FOCS19} which maintains a randomized greedy maximal independent set, or equivalently, the output of the \pivot{} algorithm in polylogarithmic time. 

    For any vertex $v$, we draw a real $\pi(v)$ from $[0, 1]$ uniformly and independently. We say $\pi(v)$ is the {\em rank} of $v$. Recall that the \pivot{} algorithm iteratively picks a pivot uniformly from the unclustered vertices and clusters it with its unclustered neighbors. Instead of doing this, we can process the vertices in the increasing order of their ranks, discarding vertices encountered that are already clustered. The resulting clustering is equivalent. We can do the same for \modifiedpivot{} as well. Namely, each iteration of the while loop in \cref{alg:modified-pivot} picks the vertex in $V$ with the smallest rank. Again, the resulting clustering is equivalent.

    \paragraph{Background on the algorithm of \cite{BehnezhadDHSS-FOCS19}:} The algorithm of \cite{BehnezhadDHSS-FOCS19}, for each vertex $v$, maintains the following data structures dynamically:
    \begin{itemize}
        \item $elim(v)$: This represents the  pivot by which vertex $v$ is clustered. If $v$ itself is a pivot, then $elim(v) = v$.
        \item $N^-{(v)} := \{ u \in N(v) \mid \pi(elim(u)) \leq \pi(elim(v))\}$: Intuitively, these are the neighbors of $v$ clustered no later than $v$. The algorithm stores $N^-(v)$ in a balanced binary search tree where each vertex $u$ is indexed by $\pi(elim(u))$.
        \item $N^+(v) := \{ u \in N(v) \mid \pi(elim(u)) \geq \pi(elim(v))\}$:  These are neighbors of $v$ clustered no sooner than $v$. The algorithm stores $N^+(v)$ in a BST indexed by the static vertex IDs.
    \end{itemize}

    \begin{lemma}[Lemma~4.1 of \cite{BehnezhadDHSS-FOCS19}]\label{lem:parametrized-ut}
        Let $\mc{A}$ be the set of vertices whose pivot changes after inserting or deleting an edge $(a, b)$.
        There is an algorithm to update all the data structures above in time
        $$
        \widetilde{O}\left(|\mc{A}| \cdot \min \left\{\Delta, \frac{1}{\min\{\pi(a), \pi(b)\}} \right\}\right).
        $$
    \end{lemma}

    Combined with the following lemma also proved in \cite{BehnezhadDHSS-FOCS19}, this implies that all the data structures can be updated in polylogarithmic time.

    \begin{lemma}[Lemma~5.1 of \cite{BehnezhadDHSS-FOCS19}]
        Let $\mc{A}$ be as in \cref{lem:parametrized-ut}. It holds for every $\lambda \in (0, 1]$ that
        $$
            \E\left[|\mc{A}| \mid \frac{1}{\min\{\pi(a), \pi(b)\}} = \lambda\right] = O(\log n).
        $$
    \end{lemma}

    These two lemmas combined, imply that the update-time is polylogarithmic in expectation.

    \paragraph{Needed modifications to maintain the output of \modifiedpivot{}.} Let us now discuss the needed modifications to maintain the output of \modifiedpivot{} also in polylogarithmic time. First, we start with the following useful claim.

    \begin{claim}\label{cl:neighborhood-counting-logtime}
        Take vertices $u$ and $v$ such that $v$ is a pivot and $\pi(elim(u)) \geq \pi(v)$. Having access to the data structures above stored by \cite{BehnezhadDHSS-FOCS19}, it is possible to determine the values of $|N(u) \cap C_v|$ and $|N(u) \Delta C_v|$ exactly in $O(\log n)$ time.
    \end{claim}
    \begin{proof}
        To see this, recall first that for each vertex $w \in C_v$, we have $elim(w) = v$. Therefore, for any edge $(u, w) \in E$, because of the assumption $\pi(elim(u)) \geq \pi(v)$, it holds that $w \in N^-(u)$. Recalling that $N^-(u)$ is indexed by the eliminator ranks, and noting that in a BST, we can count how many elements are indexed by the same value in $O(\log n)$ time, we get that we can immediately compute the value of $|N(u) \cap C_v|$ in $O(\log n)$ time. Also note that $|N(u) \Delta C_v| = d_u - |N(u) \cap C_v|$, where $d_u$ is the total number of neighbors of $u$ whose eliminator rank is at least $\pi(v)$. Such neighbors of $u$ can be both in $N^-(u)$ and $N^+(u)$. We can count the ones in $N^-(u)$ by simply using the properly indexed BST in $O(\log n)$ time, and can simply sum it up to $|N^+(u)|$ since all neighbors of $u$ in $N^+(u)$ contribute to $d_u$. This concludes the proof.
    \end{proof}

    In addition to the data structures maintained by the algorithm of \cite{BehnezhadDHSS-FOCS19}, for each vertex $u$ and each $S \in \{C, D, D', A, A'\}$, we store a pointer $I_S(u)$ which takes the value of either a vertex $v$ or $\perp$. If $I_S(u) = v$, this implies that $u \in S_v$. If $I_S(u) = \perp$, then $u \not\in S_v$ for any $v$. For instance, if $I_{D}(u) = v$, we get that $u \in D'_v$. Note that by having these pointers, we can also immediately maintain the sets $C_v, D_v, D'_v, A_v, A'_v$ for each pivot $v$. To do so, whenever $I_S(u)$ changes from $v$ to $v'$, we delete $u$ from $S_v$ and insert it to $S_{v'}$. This can be done in $O(\log n)$ time by storing these sets as BSTs.

    Below, we discuss how these data structures can be maintained in the same time as \cref{lem:parametrized-ut}.

    \begin{itemize}
        \item $I_C(u)$: Note that $I_C(u)$ is equivalent to $elim(u)$, which is already maintained by \cite{BehnezhadDHSS-FOCS19}.
        \item $I_D(u)$: Suppose that $I_D(u) = v$, i.e., $u \in D_v$. An update may change the value of $I_D(u)$ under one of these events: $(i)$ the pivot of $u$ changes, $(ii)$ some vertices leave or are added to $C_v$, changing the criteria $|N(u) \cap C_v| \leq \delta|C_v| - 1$ for $u$, or $(iii)$ an edge is inserted or deleted from $u$ to some other vertex in $C_v$. We discuss how to efficiently update $I_D(u)$ in each of these scenarios. 
        \begin{itemize}
            \item[$(i)$] Suppose that a vertex $v$ is now marked as a pivot after some update. We argue that we can identify $D_v$ in $O(|C_v| \log n)$ time. To do so, we go over all vertices of $C_v$ one by one, and apply the algorithm of \cref{cl:neighborhood-counting-logtime} on each to check whether they belong to $D_v$.  Since, from our earlier discussion, we already explicitly maintain $C_v$ which requires $\Omega(|C_v|)$ time when $v$ is marked as a pivot, this only increases the update-time by a $O(\log n)$ factor.
            \item[$(ii)$] Now suppose that a vertex $w$ is added to $C_v$. In this case, we go over all vertices of $N^+(w)$, and for each one $u$, recompute the value of $|N(u) \cap C_v|$ in $O(\log n)$ time as discussed to decide whether $I_D(u) = v$. Note that $w$ must belong to set $\mc{A}$ (defined in \cref{lem:parametrized-ut}), and its neighborhood $N^+(w)$ has size at most $O(\log n/\pi(w))$ (see Proposition~3.1 of \cite{BehnezhadDHSS-FOCS19}). Since $\pi(w) \geq \min\{\pi(a), \pi(b)\}$ where $(a, b)$ is the edge update causing this change, the total running time of this step is upper bounded by $$
            \widetilde{O}\left(|\mc{A}| \cdot \min \left\{\Delta, \frac{1}{\min\{\pi(a), \pi(b)\}} \right\}\right),
            $$
            which is also spent by the algorithm of \cite{BehnezhadDHSS-FOCS19} (\cref{lem:parametrized-ut}). The process for when a vertex $w$ is removed from $C_v$ is similar.
            \item[$(iii)$] In this case, we simply re-evaluate $|N(u) \cap C_v|$, which can be done in $O(\log n)$ using \cref{cl:neighborhood-counting-logtime}.
        \end{itemize}
        \item $I_A(u)$: To maintain $I_A(u)$, we maintain another pointer $I_A(u, v)$ for every pair of vertices $u$ and $v$ which is 1 iff $v$ is a pivot, $\pi(elim(u)) > \pi(v)$, and $|N(u) \Delta C_v| \leq \epsilon|C_v| - 1$ (where with a slight abuse of notation, $N(u)$ is the neighbors of $u$ remained in the graph at the time that $v$ is chosen as a pivot). This way, $I_A(u)$ is exactly the vertex $v$ minimizing $\pi(v)$ such that $I_A(u, v) = 1$. So let us see how we maintain $I_A(u, v)$ efficiently.

        An update may change the value of $I_A(u, v)$ under one of these events: $(i)$ whether $v$ is a pivot changes, $(ii)$ some vertices leave or are added to $C_v$, changing the criteria $|N(u) \Delta C_v| \leq \epsilon|C_v| - 1$ for $u$, or $(iii)$ an edge is inserted or deleted from $u$ to some other vertex in $C_v$. We discuss how to efficiently update $I_A(u, v)$ in each of these scenarios. 
        \begin{itemize}
            \item[$(i)$] Suppose that $v$ is marked as a pivot after an edge update. We will show how to find all vertices $w$ that satisfy $|N(u) \Delta C_v| \leq \epsilon|C_v| - 1$ in total time $O((\log^3 n)/\pi(v))$. Since we can afford to spend this much time for every vertex in $\mc{A}$ due to \cref{lem:parametrized-ut}, this will keep the update-time polylogarithmic.
            
            To do so, we subsample $\Theta(\log n)$ vertices in $C_v$ without replacement and call it $S_v$. We then take $\hat{A} = \cup_{x \in S_v} N^+(x)$. Note that we have $|N^+(x)| \leq O(\log n / \pi(x)) \leq O(\log n / \pi(v))$ for each $x \in C_v$ by (see Proposition~3.1 of \cite{BehnezhadDHSS-FOCS19}). Hence, $\hat{A}$ has size at most $O(\log^2 n / \pi(v))$. We go over all vertices $u$ in $\hat{A}$ and check, using \cref{cl:neighborhood-counting-logtime}, whether $I_A(u, v) = 1$ by spending $O(\log n)$ time.
            
            It remains to show that if $I_A(u, v) = 1$, then $u$ must belong to $\hat{A}$. Indeed, we show this holds with probability $1-1/\poly(n)$. To see this, note that $I_A(u, v) = 1$ iff $|N(u) \Delta C_v| \leq \epsilon |C_v| - 1.$ This means $u$ must be adjacent to at least a constant fraction of vertices in $C_v$. Since $S_v$ includes $\Theta(\log n)$ random samples from $C_v$, $u$ is adjacent to at least one with probability $1-1/\poly(n)$.

            \item[$(ii)$] Suppose a vertex $w$ is added to $C_v$. In this case, we go over all vertices $x$ of $N^+(w)$ and on each reevaluate whether $I_A(x, v) = 1$ in $O(\log n)$ time using \cref{cl:neighborhood-counting-logtime}. The needed running time is $O(\log n / \pi(w))$ for $w$. Similar to case $(ii)$ of updating $I_D(v)$, this, overall, takes the same time as in \cref{lem:parametrized-ut} keeping the update-time polylogarithmic.

            \item[$(iii)$] In this case, we just reevaluate $I_A(u, v)$ in $O(\log n)$ time using \cref{cl:neighborhood-counting-logtime}.
        \end{itemize} 
        \item $I_{D'}(u), I_{A'}(u)$: Note that $A'_v$ and $D'_v$ are simply random-subsamples of $A_v$ and $D_v$ respectively. Since we explicitly maintain $A_v$ and $D_v$, we can also explicitly maintain these random subsamples as efficiently, and thus can maintain $I_{D'}(u)$ and $I_{A'}(u)$ accordingly.
    \end{itemize}

    This wraps up the discussion on how we efficiently maintain our data structures. Having all the sets $C_v, D_v, D'_v, A_v, A'_v$ maintained explicitly, we can also maintain the cluster of each vertex formed by \cref{alg:modified-pivot} in polylogarithmic time, concluding the proof of \cref{thm:dynamic}.
\end{proof}